\documentclass[journal,10pt]{IEEEtran}
 
\usepackage{cuted}
\usepackage{comment}
\usepackage{gensymb}
\usepackage{graphicx}
\usepackage{caption}
\usepackage [noadjust]{cite} 
\usepackage{bm}  
\usepackage{amsmath}
\usepackage{amssymb}
\usepackage{enumerate}
\usepackage{stfloats}
\usepackage{cases}
\usepackage{epstopdf}
\usepackage{soul,color} 
\usepackage{hyperref}
\usepackage[percent]{overpic}
\usepackage{tabularx}
\usepackage[table]{xcolor}
\usepackage{multirow}
\usepackage{booktabs}  
\usepackage{tabu} 

\usepackage{algorithm}
\usepackage{algpseudocode}%

\newcommand{\M}{\mathcal{M}}

\newcolumntype{L}[1]{>{\hsize=#1\hsize\raggedright\arraybackslash}X}%
\newcolumntype{R}[1]{>{\hsize=#1\hsize\raggedleft\arraybackslash}X}%
\newcolumntype{C}[1]{>{\hsize=#1\hsize\centering\arraybackslash}X}%

\newcommand{\argmax}{\operatornamewithlimits{argmax}}

\newcommand{\hbold}{\boldsymbol{h}}
\newcommand{\Abold}{\boldsymbol{A}}

\newcommand{\npower}{\sigma^2}
\newcommand{\wbold}{\boldsymbol{w}}

\newcommand{\sbold}{\boldsymbol{s}}
\newcommand{\pbold}{\boldsymbol{p}}
\newcommand{\pibold}{\boldsymbol{\pi}}
\newcommand{\abold}{\boldsymbol{a}}
\newcommand{\thetabold}{\boldsymbol{\theta}}
\newcommand{\rbold}{\boldsymbol{r}}

\newtheorem{Corollary}{Corollary}

\newtheorem{property}{Property}

\begin{document}
	\raggedbottom
	\allowdisplaybreaks

    \title{Near-Field Spot Beamfocusing: A Correlation-Aware Transfer Learning Approach

      \thanks{This research is supported by Business Finland via project 6GBridge - Local 6G (Grant Number: 8002/31/2022), and Research Council of Finland, 6G Flagship Programme (Grant Number: 346208).  The
work of Mehdi Rasti was supported by the Research Council of Finland
Profi6 336449.
    }
}
    \author{Mohammad Amir Fallah, Mehdi Monemi, \textit{Member}, IEEE, Mehdi Rasti, \textit{Senior Member}, IEEE, Matti Latva-aho, \textit{Fellow}, IEEE
 \thanks{
  Mohammad Amir Fallah is with Department of Engineering, Payame Noor University (PNU), Tehran, Iran (email: mfallah@pnu.ac.ir)  
  \\
  Mehdi Monemi,
Mehdi Rasti, and Matti Latva-aho are with Centre
for Wireless Communications (CWC), University of Oulu, 90570 Oulu, Finland.
 (emails: mehdi.monemi@oulu.fi, mehdi.rasti@oulu.fi, matti.latva-aho@oulu.fi).
\\
}}

	\maketitle
\begin{abstract}
Three-dimensional (3D) spot beamfocusing (SBF), in contrast to conventional angular-domain beamforming, concentrates radiating power within a very small volume in both radial and angular domains in the near-field zone.
Recently the implementation of channel-state-information (CSI)-independent machine learning (ML)-based approaches have been developed for effective SBF using extremely large-scale programmable metasurface (ELPMs). These methods involve dividing the ELPMs into subarrays and independently training them with Deep Reinforcement Learning to jointly focus the beam at the desired focal point (DFP).
 This paper explores near-field SBF using ELPMs, addressing challenges associated with lengthy training times resulting from independent training of subarrays.
To achieve a faster CSI-independent solution, inspired by the correlation between the beamfocusing matrices of the subarrays, we leverage transfer learning techniques. First, we introduce a novel similarity criterion based on the phase distribution image (PDI) of subarray apertures. Then we devise a subarray policy propagation scheme that transfers the knowledge from trained to untrained subarrays. We further enhance learning by introducing quasi-liquid layers as a revised version of the adaptive policy reuse technique.  We show through simulations that the proposed scheme improves the training speed about 5 times. Furthermore, for dynamic DFP management, we devised a DFP policy blending process, which augments the convergence rate up to 8-fold.

	\end{abstract}
	\begin{keywords}
	spot beamfocusing, near-field, reinforcement learning, transfer learning,  policy propagation, policy blending, quasi-liquid layer, phase distribution image
	\end{keywords}
	
	
\thispagestyle{empty}

\section{Introduction}

\begin{figure*}[h]
		\centering
		\includegraphics [width=0.9\textwidth]{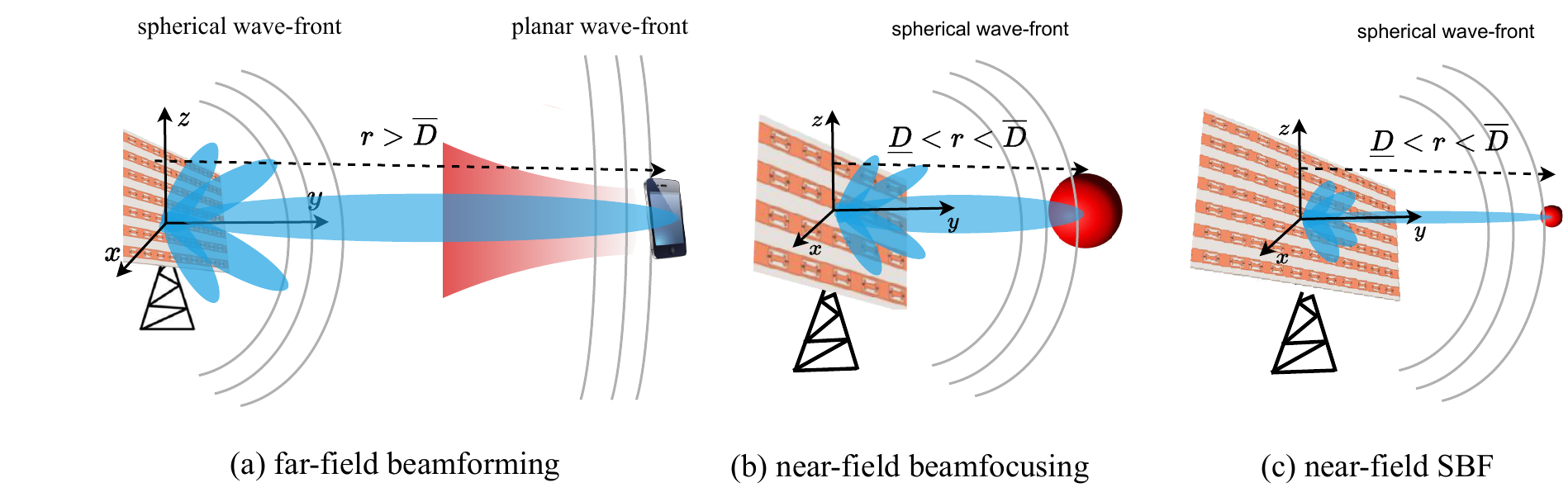} \\
		\caption{Three different types of shaping the beam profile}
		\label{fig:farfield-nearfield-SBF}
\end{figure*}

{\color{black}
Three-dimensional} (3D) beamfocusing is the technique of shaping the radiation beam profile of electromagnetic waves to concentrate the radiating power in a desired region in 3D space. This is different from the well-known far-field {\color{black}two-dimensional} (2D)  beamforming technique, which allows controlling the shape of the radiating beam in only the angular domain but not in the distance domain. 
Beamfocusing enables concentrating the radiating power by shaping the beam across all spatial dimensions, containing both the angular domains of elevation and azimuth angles and the radial distance domain. 

Electromagnetic radiating apertures, encompassing both continuous (e.g., parabolic reflectors) and discrete (e.g., planar arrays), exhibit three distinct radiation regions categorized by their distance from the aperture, as detailed in  \cite{balanis2016antenna}. The region closest to the aperture is the reactive near-field, characterized by the dominance of reactive power over radiating power. 
The furthest region, known as the far-field or Fraunhofer zone, exhibits planar propagation of radiating waves, resulting in a distance-independent radiation pattern, as depicted in  (Fig. \ref{fig:farfield-nearfield-SBF}-a) In this zone, the radiating pattern of the aperture, relative to an isotropic antenna, can only vary within the angular domain. 
The Fresnel zone, situated between the reactive near-field and the far-field zones, is characterized by a spherical wavefront. This inherent property enables beamfocusing (as illustrated in (Fig. \ref{fig:farfield-nearfield-SBF}-b), where manipulation of the phase at distinct aperture points allows concentration of radiating power into a small focal region \cite{goodman2017introduction}.
 The Fresnel zone for a continuous aperture is characterized by the region where the distance from the aperture (denoted by $r$) holds in the range of  $\underline{D} \leq r \leq$ $\overline{D}$ where $D$ is the largest dimension of the aperture,   $\lambda$ is the wavelength, $\overline{D}=2D^2/\lambda$ for the on-boresight scenario and $\underline{D}$ is limited to the space very close to the antenna aperture \cite{selvan2017fraunhofer}. The authors of \cite{10720806,monemi2024revisiting} have shown that for the off-boresight observation point, $\overline{D}$ can be extended to about 4 times that relating to the on-boresight case.

 
Near-field beamfocusing has many applications in various fields of science and technology, such as wireless communications \cite{zhang20236g,an2023toward}, wireless power transfer (WPT) \cite{demarchou2022energy,costanzo2021evolution,clerckx2018fundamentals}, medical and health \cite{khan2020wireless,tofigh2014near}, the semiconductor industry, lithography\cite{magen2021focused}, and microscopy \cite{wu2023interactive}. 
The realization of beamfocusing through conventional tools such as lenses or mirrors is widely used but entails many challenges. Resolution, aberrations, and chromatic dispersion are some of the issues that are caused by design geometry and construction materials \cite{gutierrez2021chromatic,sawant2021aberration,liang2019high}. Additionally, these tools are typically designed for a fixed {\color{black}desired focal point (DFP)}, lacking the ability for dynamic and intelligent DFP adjustment. 
To address the aforementioned limitations and enable beamfocusing with adjustable DFPs, the utilization of intelligent arrays presents a promising solution. This adjustability can be achieved through various intelligent array implementations, including conventional phased array antennas (CPAs), programmable metasurface (PMs), or holographic MIMO surfaces \cite{fu2020terahertz,gong2023holographic,an2023tutorial,yurduseven2017design,cao2020construction,kumar2022comprehensive}.

{\color{black}Spot beamfocusing} is defined by the precise control of the focal region to an exceptionally small size, thereby concentrating the bulk of radiating power into a localized, spot-like area \cite{monemi20246g,monemi2023towards}  (Fig. \ref{fig:farfield-nearfield-SBF}-c).
To achieve SBF, a high ratio of $D/\lambda$ is required \cite{goodman2017introduction}, which can be accomplished through decreasing $\lambda$ as well as increasing $D$. The former is achieved through transitioning to mm-wave/THz frequencies. The latter can be realized by employing extremely large-scale array antennas (ELAAs).
However, implementing SBF with CPAs necessitates bulky and expensive ELAAs. Additionally, holographic MIMO surfaces have limitations in fully adjustable DFP and power transmission \cite{an2023tutorial}. A viable option for SBF implementation is to use programmable metasurfaces, which provide low-cost, mass-producible, and reconfigurable ELAAs. 
Realizing SBF with ELAAs necessitates coherent signal arrival from all elements at the DFP. This necessitates phase shifters for each element, demanding the precise computation of an optimal phase-shift matrix (which is called the beamfocusing matrix in the paper).
{\color{black}
While various channel-state-information (CSI)-independent far-field beam training methods have been proposed for multi-subarray massive MIMO systems, such as the low-complexity scheme exploiting the spatial correlation between ELAA elements in \cite{wu2019low} and the fast multibeam training method for sub-connected hybrid beamforming  \cite{yang2022fast}, the application of these methods to near-field SBF are practically restrictive and hence other techniques should be developed.} The challenges stem from estimating the CSI related to a massive number of array elements and the intricate nature of near-field channel matrices (characterized by spherical wave propagation and non-sparse angular-domain representation) \cite{cui2022channel,yang2023near}. Consequently, recent research has explored machine learning (ML)-based CSI-independent SBF structures, as exemplified in \cite{monemi20246g}.


This work tackles the challenges of SBF implementation by leveraging an extremely large-scale programmable metasurface (ELPMs) and a novel deep reinforcement learning (DRL)-based algorithm. The ELPM array comprises numerous small subarrays, each equipped with a DRL agent. Each DRL agent trains the corresponding subarray independently to direct its beam toward the DFP. Subsequently, an algorithm is employed to realize constructive interference and achieve complete alignment of the subarray beams at the DFP, culminating in SBF. This modular architecture significantly reduces the complexity of the ML-based SBF algorithm.
However, inherent to their iterative nature, DRL-based algorithms necessitate a substantial number of iterations to converge upon the optimal beamfocusing matrix. Illustratively, the proposed DRL schemes in \cite{monemi20246g} require approximately 100k iterations for convergence. 
In practical applications, the optimal beamfocusing matrix for each subarray exhibits a significant degree of correlation with that of many other subarrays, as will be demonstrated in this work. 
Leveraging this inherent interdependency, the knowledge acquired during the training of one subarray can be effectively transferred to expedite the training of others. This transfer learning (TL) approach, as detailed in this paper, significantly reduces the convergence time of the DRL algorithm.
However, even minor variations in the DFP necessitate retraining all ELPM subarrays from scratch. In such scenarios, a similar TL strategy, as described earlier, can be applied to train the entire ELPM for the new DFP location, leveraging knowledge acquired from previous DFP training processes. The aforementioned cases demonstrate the necessity of applying TL techniques for SBF which is fully explored in this paper.



TL is an ML technique, in which the knowledge learned from a task is reused to boost the performance of a different but related task \cite{zhuang2020comprehensive,zhu2023transfer}. 
So far, several researches have been accomplished on the application of TL in far-field beamforming \cite{yuan2020transfer,wang2022deep,luo2023meta}. For example in \cite{yuan2020transfer}, an offline adaptive learning algorithm based on deep TL and meta-learning is proposed to achieve fast downlink beamforming adaptation in wireless communication. 
The authors of  \cite{wang2022deep} have proposed a deep transfer reinforcement learning (DTRL) method for beamforming and resource allocation in multi-cell multiple-input single-output  OFDMA systems. The DTRL method transfers the learned policy from one cell to another cell having different channel conditions and user demands, leading to the reduction of the training time. 
In \cite{luo2023meta},  a meta-critic network for reconfigurable intelligent surface (RIS)-enabled beamforming is suggested to maximize the sum rate of a multi-user system. The meta-critic network recognizes the changes in the environment and automatically updates the learning model. In addition, a stochastic explore and reload procedure is also developed to address the issue of high-dimensional action space.

Prior research has explored the application of TL for beamforming in the far-field region. However, these techniques are not directly applicable to near-field SBF through ELPMs. Currently, there exists no work in the literature considering the benefits of TL in the near-field SBF application. 
The application of TL schemes to ML problems faces some challenges, including how to devise measuring criteria to quantify the similarity or relatedness of the source and target domains or tasks; how to select the best source knowledge for transfer; and how to trade off between exploiting the source knowledge and exploring the target data. In this work, the aforementioned issues have been well addressed for the near-field SBF problem.

{\color{black}\subsection{Motivation and Contributions}
\begin{itemize}

  \item \textbf{TL-Powered Near-Field Beamfocusing:}
This paper pioneers the application of TL techniques to the domain of near-field beamfocusing and SBF. While previous research has focused on TL for far-field beamforming, our work extends its application to the near-field context. By leveraging TL, we aim to significantly accelerate the training process and enhance the efficiency of the DRL-based near-field SBF, making it a more practical and viable solution.

\item \textbf{Novel Correlation-Aware TL Approach:}
A key contribution of this paper is the introduction of a novel similarity metric, based on the \textbf{\textit{phase distribution image}} (PDI) of subarray apertures, wherein the excitation phases of array elements are considered as the pixels of a 2D image. The proposed \textit{metric} builds a correlation assessment between different PDIs, facilitating efficient knowledge transfer from the trained agents to untrained ones. Building upon this correlation assessment, we propose a novel TL-based \textbf{\textit{subarray policy propagation}} scheme, which increases the speed of finding the optimal beamfocusing matrix 4-fold compared to the case where no TL is leveraged, as evidenced through numerical results.

\item  \textbf{Quasi-Liquid Layers:}
 To improve training efficiency, we introduce a revised adaptive policy reuse technique called quasi-liquid layer (QLL). This involves adding extra layers to deep neural networks (DNNs) for actor and critic networks. The inner and middle layers are kept frozen, while the outer layers gradually become liquid. More exactly, the learning rates for the outer layers are adjusted from zero (frozen) to one (liquid). This procedure results in an extra 20\% increase in training speed compared to the subarray policy propagation scheme without QLL.

\item  \textbf{DFP Policy Blending:}
In scenarios where the DFP varies due to the movement of user equipment (UE), it is necessary to adapt the beamfocusing matrix to accommodate the changing propagation environment. Instead of training a new policy from scratch, DFP policy blending leverages the knowledge gained from previously trained policies learned for different DFP locations. The process starts by identifying a set of previously trained policies relating to old DFPs that are highly correlated with those relating to the new DFP. These correlated policies are then blended to create an initial policy for the new DFP. This blended policy significantly increases the training speed by up to 8 times compared to that when no policy blending is applied, as demonstrated  in the numerical results.
\end{itemize}
}


 \subsection{Notations:} Throughout this paper, for any matrix $\Abold$, $\Abold_{m,n}$, $\Abold^{\mathrm{T}}$, $\Abold^*$, and $\Abold^\mathrm{H}$, denote the $(m,n)$-th entry, transpose,
conjugate, and conjugate transpose
respectively. Similarly, for each vector $\abold$, $\abold_{n}$, $\abold^{\mathrm{T}}$, $\abold^*$, $\abold^\mathrm{H}$, $\|\abold\|$
denote the $n$-th entry, transpose,
conjugate, conjugate transpose, and Euclidean norm
respectively. 
Finally, {\it vec} converts a matrix into a vector, $\odot$ is the element-wise Hadamard product and $\boldsymbol{1}_N$ is a $N\times 1$  ones-vector.

\section{System Model and Problem Formulation}
\label{sec:system_model_and_problem_formulation}

\subsection{System Model}
We assume a planar ELPM with $N=N_r\times N_c$ antenna elements arranged as a rectangular phased array aperture of diameter $D$, where each antenna element has a $q$-bit quantized phase shifter to control the phase of the radiating signal. Let $\boldsymbol{r}_{ij}^a$, $\boldsymbol{r}^a$, and $\boldsymbol{r}^\mathrm{DFP}$ be respectively the location points of the antenna element $(i,j)\in\{1,2,..., N_r\}\times\{1,2,..., N_c\}$, the center of the aperture, and the DFP location. The channel gain matrix between the ELPM and DFP is  $\hbold=[h_{ij} g_{ij}^a g^{r}]^{N_r\times N_c}$ where $g_{ij}^a$ is the gain of the transmitter antenna element $n$ (in which the mutual coupling between the antenna elements is considered), $g^{r}$ is the gain of the receiver antenna, and $h_{ij}$ is the corresponding channel gain obtained as follows:
\begin{multline}
    \label{eq:hn}
    h_{ij}=\gamma\|\rbold_{ij}^a-\rbold^\mathrm{DFP} \|^{-\frac{\alpha}{2}}e^{-j\left( \frac{2\pi}{\lambda} \|\rbold_{ij}^a-\rbold^\mathrm{DFP} \| + \Delta\theta_{n0}\right) }
    +
    \\
    \gamma\sum_{l=1}^{L} 
    \left[
    \beta_{ijl} \left(d_{ijl}\right)^{-\frac{\alpha}{2}} e^{-j\left(\frac{2\pi}{\lambda}d_{ijl}+\Delta \theta_{ijl}\right)}
    \right].
\end{multline}
 where  $\gamma$ is the attenuation coefficient, $\alpha$ is the path-loss exponent, $d_{ijl}$ is the total path length of the propagated signal from the $l$'th path of antenna element $(i,j)$ toward the DFP,
$\beta_{ijl}$ is the corresponding signal attenuation coefficient due to the reflection and finally, $\Delta \theta_{ijl}$ models the corresponding phase shift initiated from the reflecting surfaces, as well as the phase mismatch due to hardware impairments.
 

Let $\wbold=[w_{ij}]_{N_r \times N_c}$ be an $N$ element \textit{beamfocusing matrix} where $w_{ij}$ is the coefficient corresponding to the antenna element in the $i$'th row and $j$'th column of the array. Each element  $w_{ij}$  is in the form of $\frac{1}{\sqrt{N}} e^{j\phi}$ where $\phi\in\phi^{\mathrm{quan}}$ is the phase-shift coefficient, in which $\phi^{\mathrm{quan}}$ is the set of $2^q$ valid phases obtained from the $q$-bit quantized phase shifter. Without loss of generality, we consider that $\wbold$ is composed of $M=M_r\times M_c$ subarrays where the set of subarrays is denoted by $\M=\{1,2,...,M\}$ and each subarray $m\in\M$ has $N'=N'_r\times N'_c$ elements. Therefore, $\wbold$ is expressed as follows:

      \begin{align}
      \label{eq:BM}
         \wbold= 
         \begin{bmatrix}
            \wbold^{(1)}  &  \wbold^{(2)}  & ...  & \wbold^{(M_c)} 
            \\            \wbold^{(M_c+1)}  &  \wbold^{(M_c+2)}  & ...  & \wbold^{(2M_c)}
            \\
            \vdots  & \vdots  & \ddots & \vdots
            \\
            \wbold^{(M_c(M_r-1)+1)}  & \wbold^{(M_c(M_r-1)+2)}   & ...  & \wbold^{(M_c M_r)}
        \end{bmatrix}
     \end{align}
where  $\wbold^{(m)}=[w^{(m)}_{ij}]_{N'_r \times N'_c}$ is the \textit{beamfocusing submatrix} for subarray $m$. By letting $s$ be the transmit signal, and $\nu$ be the additive Gausian noise, the received signal at DFP is obtained as   
   \begin{align}
x =\boldsymbol{1}^\mathrm{T}_{N_r} \left( \wbold^*\odot   \hbold \right) \boldsymbol{1}_{N_c} s+\nu
\end{align}
   
   For any given beamforming matrix $\wbold$, the received power at DFP location $\rbold^\mathrm{DFP}$ is obtained as follows: 
\begin{align}
     \label{eq:power_def}
     &  p(\wbold,\rbold^\mathrm{DFP})
    =\mathbb{E} [ x x^*]= \notag
     \\
     & \hspace{30pt}
     \boldsymbol{1}^T_{N_r} \left( \wbold^*\odot   \hbold \right)  \odot \left( \wbold\odot   \hbold^* \right)^T  \boldsymbol{1}_{N_c}{E}[s s^*] +
     \notag
     \\
     &\boldsymbol{1}^T_{N_r} \left( \wbold^*\odot   \hbold \right) \boldsymbol{1}_{N_c}{E}[s \nu^*] +
     \boldsymbol{1}^T_{N_r} \left( \hbold^*\odot   \wbold \right) \boldsymbol{1}_{N_c} {E}[s^* \nu] + \npower_{\nu}. 
 \end{align}
 in which $\npower_{\nu}$ is the noise power.

\subsection{Problem Formulation}
To formally define the SBF problem, first, we introduce the concept of beamfocusing radius (\textbf{BFR}) as presented in \cite{monemi20246g}.
  Given any beamfocusing matrix $\wbold$, DFP locator $\rbold^\mathrm{DFP}$, and some constant $0<\eta\leq 1$, the BFR denoted by $R(\wbold,\rbold^\mathrm{DFP},\eta)$ is defined as the radius of the circle $S_R$ located on a reference plane $S$ centered at DFP that contains a fraction  $\eta$ of the total radiating power passing through the reference plane $S$ . For instance, for $\eta=0.9$, it can be assumed that 90\% of the total radiating energy is confined within a circle on the reference plane with radius $R$ and center DFP.
{\color{black}$R$  is formally obtained by finding a value of  $S_R$  for which the following  equality holds \cite{monemi20246g}: 
    \begin{align}
    \label{eq:R}
        \int_{S_R} \partial \pbold(\wbold,\rbold').\widehat{\boldsymbol{a}}_nd s'
            = \eta 
            \int_{S} \partial \pbold(\wbold,\rbold').\widehat{\boldsymbol{a}}_nd s'=\eta P^T
    \end{align}
where $\partial\pbold(\wbold,\rbold')$ is the power density function corresponding to $\wbold$ the location point $\rbold'$, and $\widehat{\boldsymbol{a}}_n$ is the unit vector normal to the reference plane (focal plane), and $P^T$ is the total power passing from the focal plane.}
Now consider a sphere of radius  $R$ centered at the DFP. It is demonstrated that when the near-field effect becomes significant, the power level decreases outside the sphere in all directions, even when approaching the aperture. This suggests that the focal region in the near-field can be approximated by a sphere centered at DFP.

Based on the definition of BFR, the SBF problem aims at finding the beamfocusing matrix $\wbold$ corresponding to the minimum BFR, formally expressed as follows
 \begin{subequations}
\label{eq:opt1}
    \begin{align}
	\min_{\wbold} & \ R(\wbold,\rbold^\mathrm{DFP},\eta) 
        \\
        \label{eq:opt1_const1}
        \mathrm{subject\ to:} 
        & \ \wbold \in \mathcal{W} 
        \\
        \label{eq:opt1_phys1}
        \mathrm{phy.\ const.:}
        & \ \lVert \boldsymbol{r}^a-\boldsymbol{r}^\mathrm{DFP}\rVert\in [\underline{D},\overline{D}]
        \\
        \label{eq:opt1_phys2}
        & \  R(\wbold,\rbold^\mathrm{DFP},\eta) < BFR^{\mathrm{th}}
    \end{align}
\end{subequations}
where $\underline{D}$ and $\overline{D}$ are the Fresnel zone limits, $\mathcal{W}$ is the domain of $\wbold$ which is of the cardinality of $2^{r\times N}$, and $\lVert \boldsymbol{r}^a- \boldsymbol{r}^\mathrm{DFP}\rVert$ is the distance between the center of the aperture and DFP. 
The physical constraint \eqref{eq:opt1_phys1} guarantees that the DFP is located in the Fresnel region, and the one in \eqref{eq:opt1_phys2} ensures that the BFR is lower than a small threshold value $BFR^{\mathrm{th}}$ which is determined by the minimum size of a focal region to be regarded as a spot. To satisfy the physical constraint  \eqref{eq:opt1_phys2}, a large number of array elements is needed. This requires $N$ to be higher than a minimum number of elements $N^{\mathrm{th}}$ where $N^{\mathrm{th}}$ is a function of $BFR^{\mathrm{th}}$; the lower the $BFR^{\mathrm{th}}$, the higher the number of array elements needed.

The following property is used for proving the next Theorem.
\begin{property}
   
The normalized array responses for any beamforming matrix $\wbold$ and any given two points $\rbold_1 \neq \rbold_2$ denoted by $\abold(\rbold_1,\wbold)$ and $\abold(\rbold_2,\wbold)$ have asymptotic orthogonality in the near-field if the number of antennas $N$ is sufficiently large. This property is mathematically expressed as follows \cite{liu2023near}:
\begin{align}
    \label{eq:asymp}
    \lim_{N\rightarrow \infty} \frac{1}{N} |\abold^{\mathrm{T}} (\rbold_1,\wbold) \abold^*(\rbold_2,\wbold)|=0
\end{align}
where $\abold(\rbold,\wbold)=\mathrm{vec} (\wbold^* \odot \boldsymbol{h}) $.
\end{property}

\begin{Corollary} 
    \label{asymp_theorem}

In the asymptotic case where $N$ is sufficiently large, the BFR minimization problem \eqref{eq:opt1}  is equivalent to maximizing the focused power at the DFP, corresponding to the following optimization problem:
 \begin{subequations}
\label{eq:opt_power}
    \begin{align}
	\max_{\wbold} & \ p(\wbold,\rbold^\mathrm{DFP}) 
        \\
        \mathrm{subject\ to:} 
        & \eqref{eq:opt1_const1}
        \\
        \mathrm{phys. \ const.:} &\eqref{eq:opt1_phys1},
        \eqref{eq:opt1_phys2}
    \end{align}
\end{subequations}
{\color{black} where $p$ is formulated in \eqref{eq:power_def}.}
\end{Corollary}

\begin{proof}
From Property 1, in the asymptotic case where $N$ is sufficiently large, for the point $\rbold_1=\rbold^{\mathrm{DFP}}$, the array response $\abold(\wbold,\rbold_1)$ corresponding to maximum power at $\rbold^{\mathrm{DFP}}$ is orthogonal to the array response for any point $\rbold_2 \neq \rbold^{\mathrm{DFP}}$, leading to negligible received power at point $\rbold_2$. This means that the point $\rbold_2\neq \rbold^{\mathrm{DFP}}$ resides outside the BFR, which results in a very small BFR in the asymptotic case corresponding to the BFR minimization problem expressed in \eqref{eq:opt1}.
\end{proof}



It should be noted that solving \eqref{eq:opt_power} requires the exact CSI of all ELPM elements. Current schemes for the estimation of CSI in massive MIMO and ELPM systems in the conventional far-field propagation scenario rely basically on the assumption of sparse representation of the propagation channels \cite{zheng2022survey,cui2022channel}.


{\color{black}In the classical far-field angular-domain model, near-field channels might not appear sparse due to the spherical wavefront \cite{cui2022channel}. Recent research has shown that by transitioning to a polar-domain or a distance-parameterized angular-domain representation \cite{cui2022channel,zhang2023near}, sparse representations of near-field channels can be available. However, the utilization of this sparsity assumption is challenging for near-field SBF applications leveraging ELPMs due to the following reasons: 
(a)-{\it High-cardinality codebook space:} The works presented in \cite{cui2022channel,zhang2023near} employ a 2D distance-angular representation for a 1D linear antenna array, partitioning the angular-domain only for the elevation angle ($\theta$). In contrast, 3D SBF requires ELAAs with 2D planar antenna array comprising thousands of elements.
This necessitates a high-dimensional codebook space. 
    While incorporating distance-domain partitioning can potentially result in a sparse channel representation, it adds another dimension to the channel partitioning model leading to an extremely large cardinality of the codebook space.
(b)-{\it Impact of multipath components:} The validity of the sparsity assumption for the distance-angular channel representation proposed in  \cite{cui2022channel,zhang2023near} is contingent upon the number of multipath components. As this parameter increases, the sparsity of the channel representation gradually diminishes.}

Considering the aforementioned issues, the channel estimation for ELPM  in the near-field scenario is a challenging issue. Therefore, a pure blind CSI approach is proposed to solve equation  \eqref{eq:opt_power}. The investigation of this problem has been dealt with in \cite{monemi20246g} using an ML scheme; the idea is to divide the array into a set of subarrays, where the policy of each subarray $m$  is independently trained through a DRL optimizer agent to find optimal beamfocusing submatrix {\color{black}$\widehat{\wbold}^{(m)}$}, and then the optimal beamfocusing matrix {\color{black}$\widehat{\wbold}$} is constructed through some beam alignment mechanism. While possessing attractive benefits, the practical applicability of the described approach is limited by the slow convergence observed in the DRL algorithm. To address this limitation, we propose leveraging the inherent structural similarity existing in beamfocusing matrix/submatrices as elaborated later on. 
Consider the case where $\M^{+}$ policies, each associated with a distinct subarray, are trained independently. Due to the similarity of subarrays beamfocusing submatrices,
it is conceivable that the knowledge acquired from the trained policies $\M^{+}$ could be transferred to expedite the training of remaining $\M^{-}$ subarray policies. This approach has the potential to significantly improve overall training speed. Furthermore, consider a set $\boldsymbol{\mathcal{D}}^{+}$ encompassing all DFPs for which the ELPM has been previously trained. By exploring whether the policies associated with DFPs in $\boldsymbol{\mathcal{D}}^{+}$ can be utilized to train the ELPM for a new DFP, we aim to achieve faster convergence for dynamic DFP management. These two problems are formally defined in the following.


\subsubsection{Subarray Training Propagation}
The first part focuses on exploring how previously trained subarrays can facilitate the training of newly introduced subarrays within the ELPM. This concept is formally defined as follows.

\textbf{Problem 1 (P1)}:
For a given DFP, let {\color{black}$\widehat{\wbold}$} be the solution of the SBF problem \eqref{eq:opt_power}, and suppose $\M=\M^{+} \cup \M^{-}$, where $\M^{+}$ is the set of subarrays whose corresponding beamfocusing submatrices have been obtained (through ML schemes), and  $\M^{-}$ is the set of subarrays whose optimal beamfocusing submatrices are unknown. 
We aim at calculating the overall beamfocusing matrix {\color{black}$\widehat{\wbold}$} by finding {\color{black}$\widehat{\wbold}^{(m')}$} $, \forall m'\in\M^{-}$ through the incorporation of the known submatrices {\color{black}$\widehat{\wbold}^{(m)}$} $, \forall m\in\M^{+}$.

\subsubsection{Dynamic DFP Management}
The second part delves into the challenge of calculating the beamfocusing matrix for a new DFP. We examine how the policies trained for previously encountered DFPs can contribute to the learning process of the policy aimed at the new DFP. This exploration can be formally expressed as follows.

\textbf{Problem 2 (P2)}: Let $\boldsymbol{\mathcal{D}}^{+}$ be the set of DFPs for which the corresponding beamfocusing matrices are available (i.e., trained through ML and stored), and let $\mathcal{D}^{-}$ be a new DFP for which the beamfocusing matrix is to be calculated. We aim to find the optimal beamfocusing matrix for $\mathcal{D}^{-}$ by incorporating the ones relating to $\boldsymbol{\mathcal{D}}^{+}$.

\section{Reviewing Existing DRL-Based SBF Solution}
In what follows, first, we describe briefly the recent background DRL-based solution approach to problem \eqref{eq:opt_power} \cite{monemi20246g} whose basic concepts are also required in our work, and then we highlight the challenges of the existing solution scheme. In the subsequent section, we present new TL-based solution approaches for problems \textbf{P1} and \textbf{P2}, which overcome the drawbacks of the existing solution.

\subsection{A Review on the Distributed DRL-Based SBF Solution Approach \cite{monemi20246g}}

Consider a planar ELPM. We aim to find the optimal beamfocusing matrix {\color{black}$\widehat{\wbold}$}$\in\mathcal{W}$ to achieve maximum measured power in \eqref{eq:opt_power} (corresponding to minimum BFR in \eqref{eq:opt1})  with no prior knowledge of the ELPM elements' CSI. In the subsequent, we briefly describe the distributed DRL-based solution approach explored very recently in \cite{monemi20246g}. 
The CSI-independent SBF solution scheme proposed in \cite{monemi20246g} is based on employing a revised version of the Twin Delayed deep deterministic policy gradient (TD3) DRL-based optimizer which gets an action $\wbold^{(n)}$ for each time-step $n$, and incorporates the measured feedback power $p(\wbold^{(n)},\rbold^\mathrm{DFP})$ to generate the next time-step action $\wbold^{(n+1)}$ until the beamfocusing matrix converges to an optimal value {\color{black}$\widehat{\wbold}$}.
Finding the optimal beamfocusing matrix needs searching in an extremely large action-space domain. For example, a $60\times 60$ ELPM with 4-bit phase-shifters for each antenna element requires an action-space $\mathcal{W}$ of cardinality $16^{3600}$, which is extremely large and cannot be handled directly by conventional ML methods.  To address this drawback,  the array is divided into a set of subarrays $\M$, where each subarray $m\in\M$ is assigned with a beamfocusing submatrix  $\wbold^{(m)}$ as seen in \eqref{eq:BM}. 
In this regard, instead of directly finding the beamfocusing matrix {\color{black}$\widehat{\wbold}$}  corresponding to the maximum power $p(\widehat{\wbold},\rbold^\mathrm{DFP})$ in \eqref{eq:opt_power}, this problem is split into a set of  $M$ subarray-level optimization subproblems each solved independently. From this perspective, each subarray module is equipped with a DRL agent to find the optimal beamfocusing submatrix {\color{black}$\widehat{\wbold}^{(m)}$}corresponding to maximum received power $p_m(\Hat\wbold^{(m)},\rbold^\mathrm{DFP})$. Finally, the optimal beamfocusing submatrices $\widehat{\wbold}^{(m)},\forall m$ are aligned together through a beam alignment mechanism to form the overall SBF beamfocusing matrix {\color{black}$\widehat{\wbold}$}. 
The components of the SBF system relating to \cite{monemi20246g} are shown as the blue boxes depicted in Fig \ref{fig:structure_hardware}. 
The DRL agent of each subarray employs a TD3 algorithm to obtain the corresponding optimal beamfocusing submatrix.
TD3 is a novel model-free, off-policy DRL method designed based on the actor-critic model \cite{fujimoto2018addressing}. 
The main components of the DRL agent for the first subarray (DRL agent 1) are illustrated in Fig \ref{fig:structure_hardware}. The agents interact with the \textit{environment} and consist of \textit{actor} and \textit{critic} components as explained in the following.

\subsubsection{Environment} For each training step $n$, the environment receives the \textbf{action} $\boldsymbol{a}^{(n)}$ from the actor network. Any action taken, results in some specific observation as a \textbf{state} vector denoted by $\boldsymbol{s}^{(n)}$, and a scalar \textbf{reward} denoted by $r^{(n)}$. Here, the action is the assigned beamfocusing submatrix for subarray $m$ denoted by $\boldsymbol{a}^{(n)}=\wbold^{(m,n)}\in\mathcal{W}^{'}$. Besides, it is considered that the state for each time step $n$ is the action at the previous time step, therefore, $\boldsymbol{s}^{(n+1)}=\boldsymbol{a}^{(n)}$. The objective is to find a beamfocusing submatrix leading to the highest power value at $\rbold^\mathrm{DFP}$. For each time step $n$, the environment measures the power $p_m(\wbold^{(m,n)},\rbold^\mathrm{DFP})$ and assigns a reward as follows:

\begin{numcases}{r^{(n)}=}
\label{eq:reward}
+1, & if \resizebox{0.6\columnwidth}{!}{$p_m(\wbold^{(m,n)},\rbold^\mathrm{DFP})>p_m(\wbold^{(m,n-1)},\rbold^\mathrm{DFP})$}
\\
-1, & otherwise 
\nonumber
\end{numcases}

\subsubsection{Actor} 
The actor is a DNN with network parameters $\boldsymbol{\theta}^a$ whose input is the environment state vector $\boldsymbol{s}^{(n)}\in\mathcal{S}$, and its output is the action vector  $\mathbf{a}^{(n)}\in\mathcal{A}$ using the policy $\boldsymbol{\pi}^a(.|\thetabold^a): \mathcal{S} \rightarrow \mathcal{A}
$.
For each subarray $m$ at each iteration $n$, the input of the actor network is simply a vector containing beamfocusing submatrix elements at the corresponding iteration; therefore $\mathcal{A}=\mathcal{W}^{'}$. The actor network in a standard TD3 structure has a continuous action space; therefore a quantizer is added to the actor in order to discretize all phase elements, mapping the action phase domain from $[0,2\pi]^{N'}$ to a discrete domain corresponding to $\mathcal{W}'$.

\subsubsection{Critic}
TD3 structure is composed of two DNNs with parameters $\boldsymbol{\theta}^{Q_i},i\in\{1,2\}$, each estimating a corresponding $Q$ value. The $Q$ value is a meter used in Deep $Q$ Learning schemes that estimates how good an action is.
In the standard TD3 network, both critic DNNs take the concatenation of the state and action as the input, and generate the corresponding Q value as the output, i.e.,
\begin{align}
    \label{eq:critic_main}
    Q_i(., .|\boldsymbol{\theta}^{^{Q_i}}): \mathcal{S} \times \mathcal{A} \rightarrow \mathbb{R}, \ \forall i\in\{1,2\}
\end{align}
The overall $Q$ value is then calculated as $Q=\min{(Q_1,Q_2)}$.

\subsubsection{Agent} The agent is responsible for training the actor and critic DNNs and controlling the trend of the learning process and convergence of the training scheme. In TD3, in addition to the actor and critic DNNs, there exists a \textit{target actor} DNN denoted by $\pibold^{a,t}(.|\thetabold^{a,t})$ and two \textit{target critic} DNNs denoted by $Q_i^t(.,.|\thetabold^{Q_i,t}), \forall i\in\{1,2\}$. 
At each time step $n$, using the current observation state $\sbold^{(n)}$, the current action is selected as $\abold^{(n)}=\pibold^a(\sbold^{(n)}|\thetabold^a)+\nu^a$ where $\nu^a$ is a stochastic exploration noise which is obtained based on the noise model. Similar to the standard deep deterministic policy gradient (DDPG) model, in the conventional TD3 problem, the experience $\left(\sbold^{(n)},\abold^{(n)},r^{(n)},\sbold^{(n+1)}\right)$ is then stored and added to the experience buffer. Due to the discrete nature of the action space here, a quantization 
is added to the standard TD3 problem to convert the continuous action into a valid discrete value. By selecting a random minibatch of size $K^{mini}$ and based on the Q-learning principle, the parameters of the critic DNNs are updated at each time step $n$ through minimizing the mean square error loss function, and for each $T_1$ step, the actor parameters are updated using the sampled policy gradient to maximize the expected discounted reward. Finally, in every $T_2$ step (where $T_2>T_1$), the target actor and target critic DNNs' parameters are updated as follows:
\begin{subequations}
\label{eq:target_update}
\begin{align}
    \thetabold^{a,t}&\leftarrow\tau \thetabold^a + (1-\tau)\thetabold^{a,t}  
    \\
    \thetabold^{Q_i^t}&\leftarrow\tau \thetabold^{Q_i} + (1-\tau)\thetabold^{Q_i^t}, \ \forall i\in\{1,2\}  
\end{align}
\end{subequations}
where $\tau$ is the target smoothing factor which is a small positive value.

\subsection{Limitations and Drawbacks of the Background SBF Scheme}


 \begin{figure*}[t]
	\centering
		\includegraphics [width=512pt]{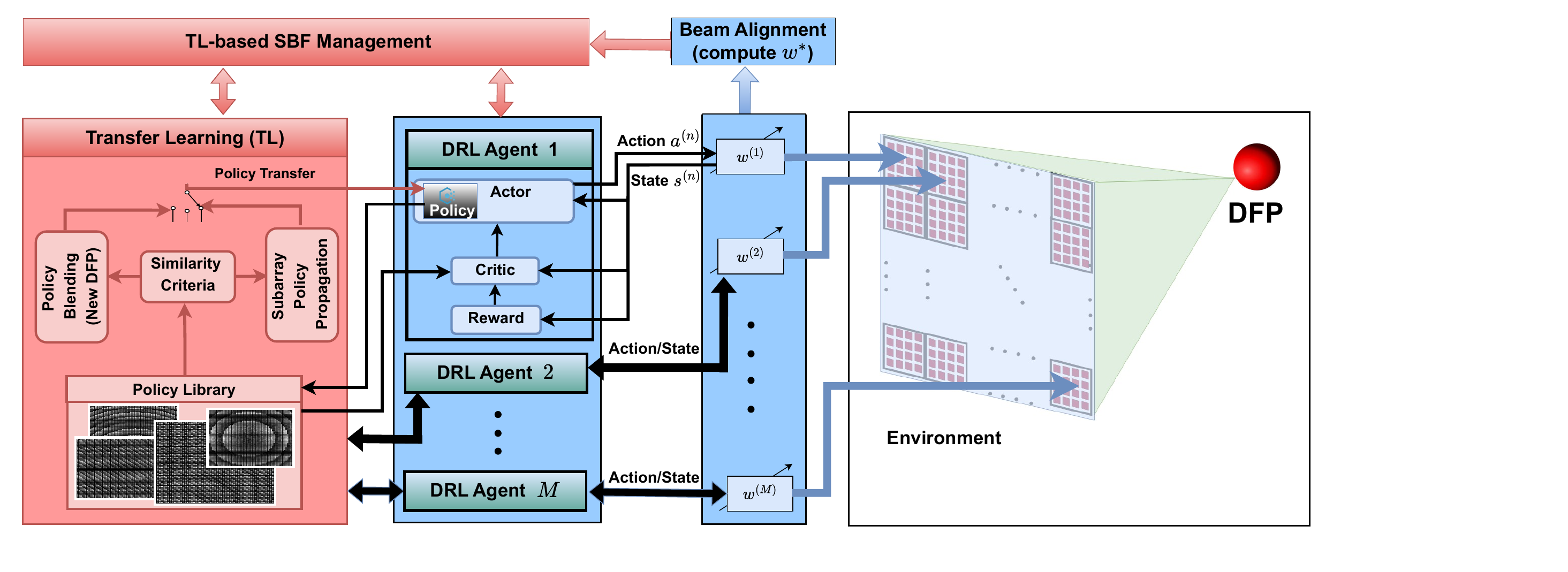} \\
		\caption{Block diagram of the proposed TL-based  SBF system.
		} 
		\label{fig:structure_hardware}
\end{figure*} 

While the distributed DRL-based approach presented in \cite{monemi20246g} remains the only proposed CSI-independent and scalable near-field SBF method for ELPMs to date, it suffers from limitations inherent to its reliance on iterative algorithms. As the study itself acknowledges, the convergence of the DRL algorithm necessitates approximately 100k iterations, posing significant challenges for practical implementation. Additionally, even minor adjustments to the DFP necessitate the complete retraining of all subarrays, rendering prior computational efforts irrelevant. This severely hinders the method's applicability in dynamic environments. One of the main reasons for this inefficiency emanates from lacking the knowledge transfer between the agents of different subarrays. 
Each subarray agent operates independently, failing to integrate prior policies or existing data from previously trained subarrays. This approach is replicated for distinct DFPs, resulting in repetitive learning from scratch without capitalizing on previously acquired knowledge. This limitation impedes the overall efficiency of the proposed DRL algorithm. Notably, as demonstrated in the subsequent section, the outputs of the optimization agents (i.e., the beamforming submatrices) exhibit significant similarities in numerous cases. These observed similarities motivate the exploration of TL schemes to leverage the knowledge of trained agents, thereby accelerating the overall learning process


\section{Transfer Learning Solution Approach}
In this section, first, we introduce TL schemes applicable to the SBF problems  \textbf{P1} and \textbf{P2}, and then characterize the requirements for applying TL schemes to the stated problems.

\subsection{Transfer Learning: How It Contributes to  the SBF Solution}
To apply TL to the DRL-based SBF solution scheme, it is important to identify the relevant knowledge to transfer as well as to choose the optimal transfer technique that is well suited to our problem. In the stated SBF problem, the source and target tasks are finding the beamfocusing submatrices, and thus, both tasks have the same state and action spaces, equal to the $\mathcal{W'}$. Additionally, the reward functions for both source and target tasks are identical. Considering these factors, two general methods for choosing the type of transfer learning technique are suggested to solve problems \textbf{P1} and \textbf{P2}.

\subsubsection{Adaptive Policy Reuse}
The conventional policy reuse is applied to the case when the source and target tasks are identical  \cite{joshi2021adaptive}. In practice, however, the source and target tasks are not quite the same in many applications. Adaptive policy reuse is an intelligent version of conventional policy reuse, where the source policy is transferred to the target domain and updated according to the target task’s dynamics. 
This can improve the performance and efficiency of the learning process in the target task by leveraging the knowledge gained from both tasks, while also allowing the agent to learn from its own experience and adjust its behavior accordingly.  This technique is leveraged for solving \textbf{P1}.

\subsubsection{Policy Blending} 
Policy blending is a technique that combines multiple policies learned from different tasks and produces a new policy that can perform well on a new task. Policy blending can be useful when there exists a set of source tasks that are somehow related to the target task. We employ this technique in solving \textbf{P2}.

In this work, we have leveraged both techniques as shown through red boxes depicted in Fig. \ref{fig:structure_hardware}. Further details on these techniques will be provided in the following subsection and throughout Section \ref{sec:TLSolution}. 
Across both proposed policy transfer techniques, the selection of the most suitable source policy hinges upon identifying the highest degree of similarity between the source and target policies. 
If the similarity exists to at least some minimum desired threshold, the source policy can be transferred to the target domain with some modification or fine-tuning for optimal adaptation. This is elaborated in the following subsection.

\subsection{ Similarity Criterion Based on the Concept of Phase Distribution Image}
The distributed DRL-based algorithm in \cite{monemi20246g} aims to find the optimal 2D  phase distribution on the ELPM aperture (optimal beamfocusing matrix {\color{black}$\widehat{\wbold}$}), which can be likened to a 2D grayscale image where each pixel has $2^q$ colors. Each pixel of this image corresponds to an antenna element of the ELPM, and each color corresponds to a quantized phase shift value. In the proposed distributed model, the beamfocusing matrix of ELPM can be considered as a phase distribution image (\textbf{PDI}), which is divided into a set of sub-images (subarray PDIs), where, the SBF algorithm finds the optimal PDI of each subarray module (corresponding to the optimal beamfocusing submatrix $\widehat{\wbold}^{(m)},\forall m\in\M$), to form the whole optimal matrix {\color{black}$\widehat{\wbold}$}.
Since \textit{PDI} is simply interpreting the \textit{beamfocusing matrix/submatrix} as an image, from now on, we may use these terms interchangeably.
For a given DFP, the locus of elements within the ELPM that receive a coherent signal from the DFP forms a series of co-focal ellipses/circles  \cite{guo2002fresnel}, as demonstrated in Fig. \ref{fig:multi_PDI}. This implies that all ELPM elements along this locus share the same phase value. As a consequence, the PDI of the ELPM exhibits a co-focal ellipses/circles shape. 
Besides, if multipath is involved, the PDI of the ELPM can be considered as the superposition of some co-focal ellipses/circles with shifted foci, which has a kind of quasi-symmetrical pattern. Therefore, it is expected that the PDI of some subarrays exhibit a high degree of similarity and even redundancy.

\begin{figure}
            \centering
		\includegraphics [width=200pt]{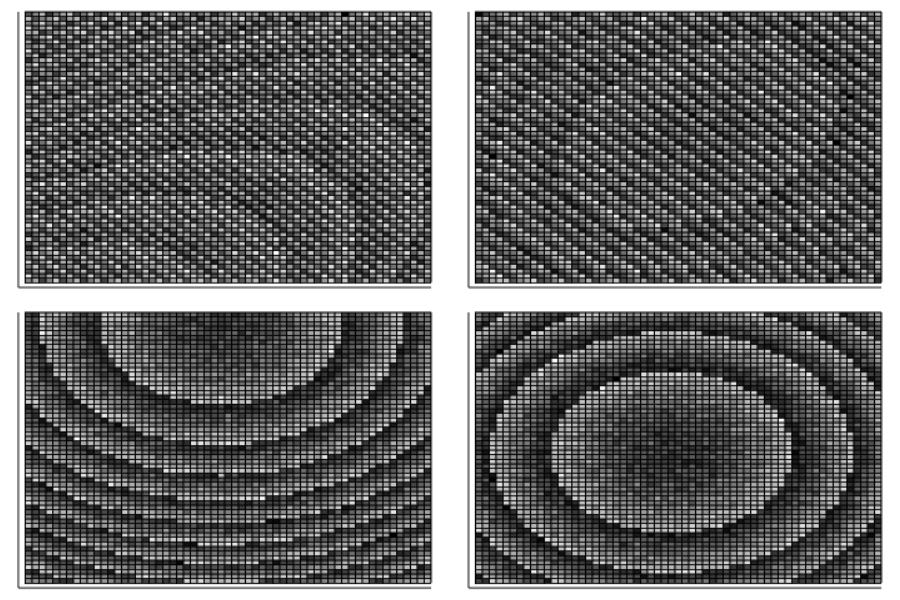} \\
		\caption{Phase Distribution Image (PDI) of the ELPM for SBF at different DFPs.
		} 
		\label{fig:multi_PDI}
\end{figure}

Two subarray PDIs may have similar patterns that are rotated relative to each other due to circular and oval patterns of ELPM's PDI, or become similar by adding a color bias value corresponding to the contrast difference between the patterns of the PDIs. In addition, all subarray PDIs have been supposed to be of the same number of pixels, and thus no scaling change is required here. Consequently, we look for a similarity criterion for identifying subarrays whose PDIs are of high similarity by considering the rotation or color bias of the images. 
Given the aforementioned conditions, we may consider the similarity criterion as the 2D Pearson cross-correlation coefficient of the PDIs, which satisfies all the above conditions and gives a value between 0 and 1 as the correlation level of two PDIs.

 The conventional Pearson cross-correlation does not take into account the {\it circularity} of the original phases corresponding to the color levels of the pixels of the PDI. To elaborate more, 
 the PDI is a gray-scale image, wherein the color of each pixel is the phase of the corresponding element ranging from 0 to 2$\pi$ (for the ideal non-quantized case). However, there is one important difference between an ordinary image and that obtained from a PDI. In an ordinary image, there exists \textbf{maximum contrast} between the color value 0, corresponding to the black color, and 255, corresponding to the white color; however, for the PDI, both colors 0 and 255 have \textbf{minimum contrast} since they are mapped to the circularly equivalent phase values of 0 and $2\pi$ respectively. In this regard, instead of the original 2D Pearson correlation metric, we employ the circular form of this metric.   Let $\wbold^{(m^s)}$ and $\wbold^{(m^t)}$ be the PDIs corresponding to the student subarray $m^{s}\in\M$ and teacher subarray $m^{t}\in\M$ respectively. 
The \textit{circular 2D Pearson's normalized cross-correlation coefficient} \cite{jammalamadaka2001topics} for subarrays $m^{(s)}$ and $m^{(t)}$ is defined as follows:
\begin{figure}
		\centering
		\includegraphics [width=200pt]{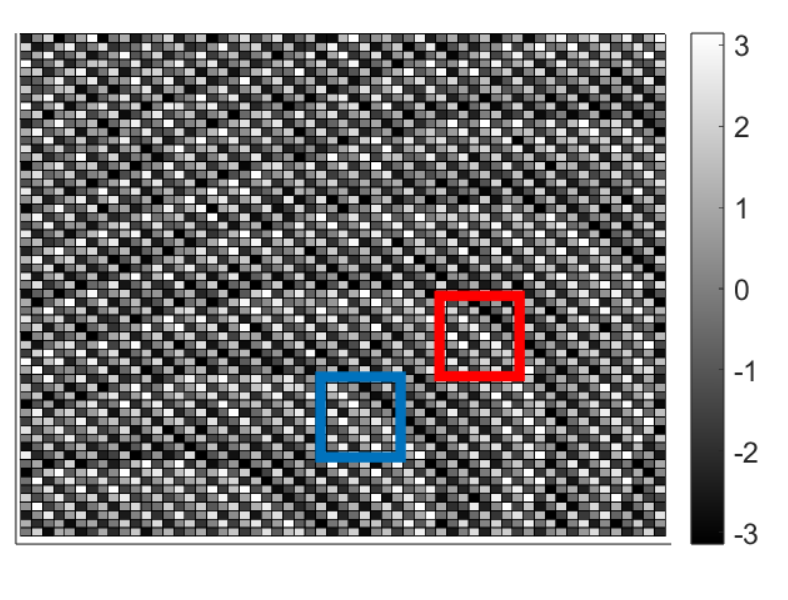} \\
		\caption{Two highly correlated regions in a typical ELPM PDI.
		} 
		\label{fig:Correlated region}
\end{figure} 
\begin{multline}
\label{eq:circular_pearson}
    C(\wbold^{(m^s)},\wbold^{(m^t)})
    =
    \\
\frac
{
\displaystyle \sum\limits_{i,j} \sin\left(\theta^{(m^s)}_{i,j}-\overline{\theta}^{(m^s)}\right)
\sin\left(\theta^{(m^t)}_{i,j}-\overline{\theta}^{(m^t)}\right)
}
{
\sqrt{
\displaystyle 
\sum\limits_{i,j} \sin\left(\theta^{(m^s)}_{i,j}-\overline{\theta}^{(m^s)}\right)^2
\sum\limits_{i,j} 
\sin\left(\theta^{(m^t)}_{i,j}-\overline{\theta}^{(m^t)}\right)^2
}
}
\end{multline}
where $\theta^{(m)}_{ij}= \arg \left(w^{(m)}_{ij}\right)$, and {\ $\overline{\theta}^{m}=\arg \left(\sum{w^{(m)}_{ij}}\right)$}.

The correlation metric expressed in \eqref{eq:circular_pearson}, considers the bias offset, as well as the circular periodicity of the color value. To further incorporate the rotation of the PDIs in the calculation of the correlation metric, we evaluate the \textit{effective circular cross-correlation coefficient} (ECC) between PDIs relating to subarrays $m^s$ and $m^t$ as follows: 
\begin{align}
\label{eq:ECC}
    C^*(\wbold^{(m^s)},\wbold^{(m^t)})=\max_{\theta \in \Theta}{C(\wbold^{(m^s)},\wbold^{(m^t)}_{\theta})}
\end{align}
 where $\Theta$ is the set of multiple evenly separated angles in the interval $[0,2\pi]$, and $\wbold^{(m)}_{\theta}$ denotes the rotation of $\wbold^{(m)}$ by $\theta$.

Fig. \ref{fig:Correlated region} illustrates two subarray PDIs of high similarity whose correlation can be evaluated through the proposed metric. As seen, the PDIs within the blue and red boxes have patterns with high similarity, due to the repeated oval shapes in the ELPM PDI and produce a high cross-correlation coefficient of about 0.9. This means that the policy trained on either of these subarrays in the SBF process can act as a teacher policy for the other, leading to a great improvement in the learning speed of the student subarray's agent.

\section{Proposed TL-Based Techniques to Solve Formulated SBF Pproblems}
\label{sec:TLSolution}
In this section, we introduce the concept of subarray policy propagation, drawing upon the adaptive policy reuse technique. This approach leverages previously trained policies for distinct subarrays within the ELPM to expedite the training of remaining subarrays, directly tackling problem \textbf{P1}. Subsequently, we propose the idea of policy blending. This method utilizes a set of pre-trained policies associated with different DFPs to construct an optimal policy for a new DFP, effectively solving problem \textbf{P2}. Through these advancements, we aim to significantly enhance the training efficiency and adaptability of the DRL-based near-field SBF approach. 

\subsection{Solution to P1 Through Adaptive Policy Reuse} 
\subsubsection{Subarray Policy Propagation} 

The SBF solution scheme proposed by \cite{monemi20246g} considers each subarray as an independent entity for training purposes. However, based on the preceding section's analysis of PDI similarity, we can leverage the policies trained for some subarrays to facilitate the training process of others. This approach exploits the inherent similarity between the subarrays’ policies and allows us to tackle problem \textbf{P1} according to Algorithm 1 as described in the following.

At the initialization phase, we assume that all subarrays of the ELPM constitute the set of students $\M^s\leftarrow\M$, and there exists no pre-trained teacher subarray, i.e., $\M^t\leftarrow\O$. To be able to transfer the policies, we need at least $K_1\geq 1$ trained subarrays' policies ($|\M^t| \geq K_1$). 
{\color{black}
In our simulations, we have selected $K_1= 1$. This choice accounts for the maximum potential benefit of any possible teacher candidate at the initialization phase, even when the set of teacher candidates is of a small cardinality.}
Following the initialization phase, a teacher policy library is established. This library’s knowledge is dynamically applied through policy reuse techniques to efficiently train student policies. The library grows dynamically and becomes enriched over time. This process is achieved in the Main Procedure in lines 3-21 of the algorithm 1.  The policy reuse technique requires finding the best source policy for the target task. We apply the similarity criterion on PDIs to find the best teacher policy for training a specific student subarray policy. 
In doing so, in steps 4-5, first, a student candidate $m^s$ is selected from $\M^s$ as some neighbor of one of the teacher candidates $\M^t$, and then the candidate teachers  $\mathcal{M}^{t,cnd}$ is selected as the first $K_2$ trained agents of subarrays with  geometrical distance from the student subarray location in the ELPM.
To find the best teacher candidate for $m^s$, for each $m\in\M^{t,cnd}$, first we start a test mode by temporarily transferring the policy $\pi^{(m)}$ to $\pi^{(m^s)}$ and then begin a temporary training phase with a very limited number of iterations for subarray $m^s$ leading to the temporary PDI, $\wbold^{(m^{s,t})}$.  
Here, we consider that the temporary phase learning rate $l^{r,t}$, and the exploration noise decay rate $\alpha^{xp,t}$ are much higher respectively than those relating to the final training phase ($l^{r}$ and $\alpha^{xp}$). This ensures that the estimated results are obtained from the test mode in very few numbers of iterations. 
From all teacher candidates, we select the one corresponding to the highest ECC according to step 14 of algorithm 1. The search for the optimal teacher candidate can be terminated if the ECC value exceeds a certain threshold. In that case, the current candidate is regarded as the best teacher. This is seen in steps 10-12 of the algorithm 1. After selecting the best teacher, we use the teacher policy for training the student subarray.
Direct policy reuse, where the identical policy from a teacher policy is mechanically applied to a student policy without adjustments, is demonstrably suboptimal. This arises due to the inherent non-equivalence of the source and target tasks' optimal beamfocusing submatrices. This disparity stems from the varying physical locations of the subarrays within the ELPM, necessitating tailored policies for optimal performance.  
Therefore, we use a revised adaptive policy reuse technique called QLLs, which is described in  Section \ref{sec:QLL}. At this stage, we employ the teacher policy $\pi^{(m^t)}$ for training the student policy  $\pi^{(m^s)}$ according to step 16 of the algorithm 1.
Then, the trained student subarray $m_s$ is removed from $\M^s$ and added to the set of possible teacher candidates $\M^t$. This enriches the policy library and increases the flexibility of finding more appropriate teachers for training the rest of the students. Since this learning process gradually propagates throughout the ELPM subarrays, we have entitled it as \textit{subarray policy propagation}.
The solution of \textbf{P1} is finally obtained by constituting the beamfocusing matrix {\color{black}$\widehat{\wbold}$} in \eqref{eq:BM} using $\widehat{\wbold}^{(m)},\forall m\in\M$.
The proposed subarray policy propagation algorithm accomplishes the training process for all $M$ subarrays of ELPM, with a much higher convergence speed than that achieved in \cite{monemi20246g}. More specifically, we will show in the numerical results that using the subarray policy propagation technique enhances the average training speed of subarrays by about 4 times.



\begin{algorithm}[t]
		\caption{\small\!: Subarray policy propagation}
		\begin{algorithmic}[1]
            \Statex \textbf{Variables:} Let $\M^s$ and $\M^t$ be the set of candidate students and teachers, $\alpha^{xp}$ be the initial exploration noise decay rate, $l^r$ be learning rate, and  $\alpha^{xp,t}$ and $l^{r,t}$ be those relating to the teacher selection in test mode respectively, where $\alpha^{xp,t} \gg \alpha^{xp}$ and $l^{r,t} \gg l^r$, and $C^\mathrm{th_1}$, and $C^\mathrm{th_2}$ two correlation thresholds where $C^\mathrm{th_2}>C^\mathrm{th_1}$.
            \Statex \textbf{Initialization:}
            \State Set $\M^s\leftarrow\M$, $\M^t\leftarrow\O$.
            \State Train a set of random $K_1 \geq 1$ subarray policies without employing TL and add the indices of the corresponding subarrays to $\M^t$.
            \Statex
            \textbf{Main Procedure:}
			\While{ $\M^s \neq \O$}
   
                \State Select a  candidate $m^s$ from $\M^s$ as some neighbor of one of the teacher candidates $\M^t$.
                
                \State Select $\M^{t,cnd}\subset\M^t$ as the first $K_2$ teacher candidate neighbors  of subarray $m^s$, whose has minimum geometrical distance to subarray $m^s$.
                \For {\textbf{each} teacher candidate $m\in \M^{t,cnd}$}  
                \State Transfer the policy of $m$ to $m^s$ $\left(\pi^{(m^s)} \leftarrow \pi^{(m)}\right)$.
                \State Start a very limited number of training iterations for subarray $m^s$ using $\alpha^{ex,t}$ and $l^{r,t}$, leading to $\wbold^{(m^{s,t})}$.
                \State Calculate $C^{(m)}= C^*(\wbold^{(m^{s,t})},\wbold^{(m)})$ from \eqref{eq:ECC}.
                \If{$C^{(m)}>C^\mathrm{th_1}$}
                    \State $m^t\leftarrow m$; Go to step 15.
                \EndIf
                \EndFor
                \State Set $m^t=\argmax_{m \in \M^{t,cnd}}\{C^{(m)}\}$.
                \If {$C^{(m^t)}>C^\mathrm{th_2}$}
			\State Transfer the policy from subarray $m^t$ to $m^s$ (i.e.,  $\pi^{(m^s)} \leftarrow \pi^{(m^t)}$), and train $\pi^{(m^s)}$ using QLL with parameters $\alpha^{ex}$ and $l^r$.
            \Else
                \State Train subarray $m^s$ without using TL.
            \EndIf
             \State Set $\M^t\leftarrow \M^t \cup \{m^s\}$, $\M^s \leftarrow \M^s \setminus \{m^s\}$.
			\EndWhile
	\State Constitute {\color{black}$\widehat{\wbold}$} according to \eqref{eq:BM} using $\widehat{\wbold}^{(m)},\forall m\in\M$.	\end{algorithmic}
	\end{algorithm}


\subsubsection{Quasi-Liquid Layers (QLLs)} 
\label{sec:QLL}

In the adaptive policy reuse method, the lower layers of NNs (those closer to the input layer) contain the history of previous training; they undergo minor changes at the beginning of the learning process. In contrast, the upper layers learn more complex patterns and have more flexibility to adapt to new tasks. This justifies the implementation of the popular approach for the adaptive policy reuse method, where the fine-tuning layer is employed in a hard-switched manner. In general, the learning rate of the outer layer is set to a non-zero value, whereas other layers are frozen by having the learning rate equal to zero.
Although this technique results in satisfactory responses in most scenarios, a better response might be obtained if we consider greater freedom by changing the hard-switching between frozen and non-frozen layers into a soft-switching scheme.
In our proposed adaptive reuse technique, we have considered liquidation of the upper layers by increasing the learning rate of these layers in a multi-step manner. We call this technique as quasi-liquid layers (QLLs). Let a DNN of the DRL originally consist of $n_1$ layers, and we add $n_2$ fine-tuning layers, resulting in a total number of layers $n_t=n_1+n_2$. We consider that the learning rate of each layer $n$ is obtained from the following relation:
\begin{align}
    \label{eq:36743}
    l^r_n=l^r_{n_t} \left[
    u(n_0)-u(n_t+1)\right] g(n)
\end{align}
where $l^r_{n_t}$ is the learning rate of the final layer, $u(.)$ is the unit step function, $n_0$ is the first non-frozen layer ($n_0\leq n_1$), and $g(n)\in(0,1]$ is an increasing function having the property $g(n_t)=1$. 
It is seen from \eqref{eq:36743} that the DNN is frozen for $n<n_0$, 
and for $n\geq n_0$, the layers are gradually turned into a liquid state by getting a higher learning rate, and finally the maximum learning rate $l^r_{n_t}$ is applied to the outer layer $n_t$. It should be noted that the commonly hard-switching fine-tune layer idea is a special case of QLL wherein $n_t=n_0$.

\subsection{Solution to P2 Through Policy Blending}

\subsubsection{Policy Blending Procedure for Dynamic DFP Management}
When the DFP changes due to UE movement, the DRL agent updates the policies of all subarrays to refocus the beam at the new location. These updated policies can be incorporated into a policy library, enriching its knowledge base and contributing to the training process for future DFPs.
This approach draws inspiration from the subarray policy propagation scenario, where policies from trained subarrays accelerate the training of others. Similarly, for dynamic DFPs, it is recommended to leverage previously acquired knowledge from the policy library instead of independent training for each new DFP. In this context, the teacher policies associated with past DFPs can be stored in the library and collectively utilized to train the entire ELPM array. 
In this context, as expressed in Problem \textbf{P2}, suppose that the ELPM has already been trained for a set of DFPs  $\boldsymbol{\mathcal{D}}^+$ corresponding to optimal policies $\boldsymbol{\pi}^{+}$, and the UE moves to a new location $\mathcal{D}^{-}$ wherein the optimal beamfocusing matrix is to be calculated according to a policy denoted by $\boldsymbol{\pi}^{-}$. To do so, we propose the {\it policy blending} technique. This method employs multiple policies (referred to as \textbf{components}) to generate a target policy output (referred to as \textbf{blend}) by applying a function (referred to as \textbf{rule}). For this scenario, the selected rule is the linear scalarization policy blending method. This method computes the weighted average of the component policies as follows:
\begin{align}
    \label{eq:blending}
    \boldsymbol{\pi}_{\mathrm{blend}}=\sum_{i=1}^{K}{\beta_{i}\boldsymbol{\pi}_{i}},\ \  \ \ \textrm{where}\ \sum_{i=1}^{K}{\beta_{i}}=1
\end{align}
where $\boldsymbol{\pi}_i=\{ \pi_i^{(1)}, \pi_i^{(2)}, ..., \pi_i^{(M)} \}$ is the policies relating to the $i$'th member of the $K$-component policy set $\boldsymbol{\pi}_{\mathrm{cmp}}$, and $\beta_i$ is the corresponding weighting coefficient; A higher weighting coefficient should be assigned to a teacher policy if it is expected to have a greater impact on the blend policy. The overall policy blending expressed in \eqref{eq:blending} can be split into the blending of the policies of the corresponding subarrays as represented in the following: 
\begin{align}
    \label{eq:blending_sub}
    \pi_{\mathrm{blend}}^{(m)}=\sum_{i=1}^{K}{\beta_{i}\pi^{(m)}_{i}}, \forall m\in\M
\end{align}
where $\pi_i^{(m)}$ is the trained policy of the $i$'th element of $\boldsymbol{\pi}_{\mathrm{cmp}}$ relating subarray $m$.
The overall policy blending procedure is presented in Algorithm 2.  
The starting point for training $\boldsymbol{\pi}^{-}$ is to constitute $\boldsymbol{\pi}_{\mathrm{cmp}}$ by selecting a limited number of proper component policies  from the policy library $\boldsymbol{\pi}^{+}$ according to some strategy which is pointed out in Section \ref{rmk:1}.
To adjust the blending  coefficients, we set a pre-training phase to estimate the effect of each component policy $\boldsymbol{\pi}_{\mathrm{cmp}}$ in the optimal policy value $\widetilde{\boldsymbol{\pi}}^{-}$. In the pre-training phase, we initially transfer each component policy $\boldsymbol{\pi}_i$ to $\boldsymbol{\pi}^{-}$, and start a very limited number of training iterations for $\boldsymbol{\pi}^{-}$ with a high learning rate and low exploration noise. 
Let  $p_{\boldsymbol{\pi}\rightarrow \boldsymbol{\pi}^{-}}$ denote the power level measured at the final step of the pre-training phase when a given policy $\boldsymbol{\pi}$ has been initially transferred to $\boldsymbol{\pi}^{-}$. Considering that the objective is to achieve SBF with a high power density at the DFP, the probability of higher similarity between  $\boldsymbol{\pi}$ and $\widetilde{\boldsymbol{\pi}}^{-}$ potentially increases if the power level $p_{\boldsymbol{\pi}\rightarrow \boldsymbol{\pi}^{-}}$ is raised. Therefore we calculate the weighting coefficient $\beta_i$ corresponding to each component policy $\boldsymbol{\pi}_{\mathrm{cmp}}$ as follows:
\begin{align}
    \label{eq:beta}
    \beta_i=
    \dfrac{
    p_{\boldsymbol{\pi}_i\rightarrow \boldsymbol{\pi}^{-}}
    }
    {
    \sum_{j=1}^{K} p_{\boldsymbol{\pi}_j\rightarrow \boldsymbol{\boldsymbol{\pi}}^{-}}
    }   
\end{align}

Once, all blending coefficients are calculated in the pre-training phase, the blending policy can be established according to  \eqref{eq:blending} for each subarray $m\in\M$. At this step, all subarrays can be trained to obtain the optimal policy $\widetilde{\boldsymbol{\pi}}^{-}$. Finally, the optimal policy is transferred to the policy library to serve as a candidate component in policy training for future DFPs.

\subsubsection{How to Select Policy Components}
\label{rmk:1}
Depending on accessible technologies and available parameters, different strategies might be adopted to select the set of component policies $\boldsymbol{\pi}_{\mathrm{cmp}}$ (corresponding to the focal points $\boldsymbol{\mathcal{D}}_{\mathrm{cmp}}$) from the policy library. For example, if the location of the focal points in $\boldsymbol{\mathcal{D}}^{+}$ is available through localization schemes, we can take up to $K$ focal points with the closest distance to $\mathcal{D}^{-}$,  as long as each distance is less than a minimum threshold value; this raises the probability of higher similarity of the PDIs of  $\boldsymbol{\mathcal{D}}_{\mathrm{cmp}}$ with that of  $\mathcal{D}^{-}$. For the case where the exact estimation of locations is not available, we may simply choose the last $K$ policies recently added to the library. 
This scheme is applicable   when the duration for the DFP to transition from one position to the subsequent one is not less than the time required to train the subarray policies at each respective DFP.


\begin{algorithm}[t]
		\caption{\small\!: Policy blending for a new DFP}
		\begin{algorithmic}[1]
                     \Statex \textbf{Variables:} Let $\boldsymbol{\mathcal{D}}^+$ be the set of focal points for which the ELPM (i.e., all subarrays) DRL policies have been trained, and $\mathcal{D}^{-}$  be a new DFP in the 3D space; Let $\boldsymbol{\pi}^{+}$ and $\boldsymbol{\pi}^{-}$ be the set of policies corresponding to $\boldsymbol{\mathcal{D}}^{+}$ and $\mathcal{D}^{-}$ respectively.  
                \Statex \textbf{Main Procedure:}
                
			\State Select $\boldsymbol{\pi}_{\mathrm{cmp}}$ as the set of $K$-component policies that are to be incorporated in the training of $\boldsymbol{\pi}^{-}$, where $\boldsymbol{\pi}_{\mathrm{cmp}}\subset\boldsymbol{\pi}^+$ is chosen according to some strategy as pointed out in Section \ref{rmk:1}.
            \State For each $\mathcal{\boldsymbol{\pi}}_i \in \boldsymbol{\mathcal{\pi}}_{\mathrm{cmp}}$, set a pre-training phase by initially transferring the trained policy of $\mathcal{\boldsymbol{\pi}}_i$ to $\mathcal{\boldsymbol{\pi}^{-}}$, and then run a very limited number of iterations with a high learning rate and low exploration noise.
            
            \State For each $\boldsymbol{\pi}_i \in \boldsymbol{\pi}_{\mathrm{cmp}}$, calculate the corresponding weighting coefficient $\beta_i$  according to \eqref{eq:beta}.
            
            \State For each $m\in\M$, constitute the ELPM DRL policy  ${\pi}_{\mathrm{blend}}^{(m)}$ 
            according to the blending rule \eqref{eq:blending},
            and start training  
            all subarrays policies to find the optimal policy ${\widetilde{\boldsymbol{\pi}}}^{-}$.
            \State Set $\boldsymbol{\pi}^{+}\leftarrow \boldsymbol{\pi}^{+} \cup \{{\widetilde{\boldsymbol{\pi}}}^{-} \}$.
		\end{algorithmic}
	\end{algorithm}


\begin{table*}
		\centering
		\caption{Simulation Parameters}
		\begin{tabular}{|l|l|l|l|}
	\hline
		\textbf{Parameter}
        & 
        \textbf{Description}
        &
        \textbf{Parameter}  
        &
        \textbf{Description} 
        \\
        \hline
        Frequency
        &
        $28$ GHz
        &
        Path-loss exponent ($\alpha$)
        &
        $2.7$
        \\
        Reflection coefficient ($\beta_{nl}, \forall n,l$) 
        & 
        $0.1$
        &
        PM subarray elements	  &
        $6\times 6$
        \\
        ELPM number of modules
	&
        100
        &
        Room Dimensions
        & 
        $4 \times 4 \times 3\ \mathrm{m}^3$                  
		\\
        Exploration noise variance ($\nu^a$)
        & 
        $0.5$
        &
        Target policy variance
&
        $0.1$
\\
        Target policy decay rate&
       $10^{-4}$
&
                $K$&
        3\\	
        $(\overline{\nu}^{xp},\underline{\nu}^{xp})$
& 
                $10^{-5}$
&
        $(T_1,T_2)$
        &
        (1,3) 
       \\
                ($C^\mathrm{th_1}$,$C^\mathrm{th_2}$)&
                (0.5,0.9)&
        $(\alpha^{xp} , l^{r})$
        &
        ($10^{-5}$ , $10^{-3}$)\\

           $K_1$
           &
            1
       &
         $(\alpha^{xp,t} , l^{r,t})$
          &
        ($5\times10^{-4}$ , $10^{-2}$)\\
        ($n_0,n_t$)
        &
        (4,8)
        &
        ($l_4^{r},l_6^{r},l_8^{r}$)
        &
       ($0.5\times10^{-4}$ , $0.6\times10^{-4}$ , $0.7\times10^{-4}$)\\
        \hline
        \end{tabular}
        \label{tbl:simulation_params}
\end{table*}
 
 \section{Numerical Results}
In this section, we investigate how TL impacts the performance of SBF for different simulation scenarios. We have considered an ELPM consisting of $60\times 60=3600$ antenna elements. The ELPM is composed of $M=10\times10=100$ subarrays, each having $N'=6\times 6=36$ elements. The horizontal and vertical distance between adjacent ELPM elements is $\lambda/2$.
Unless explicitly specified otherwise, the default simulation parameters are expressed in table \ref{tbl:simulation_params}.
Same as \cite{monemi20246g}, we have considered a multipath channel model in a $4\times4\times 3 \ \mathrm{m}^3$ room. The DFP is located in $(1,1.5,1.4)$ m, and the location of ELPM bottom-left corner is $(1,0,1.5)$ m.
The DRL for each subarray module consists of one actor DNN and two critics DNNs. The actor DNN starts with a normalization input layer followed by a $16N'$-neuron fully connected layer (FCL), a Rectified Linear Unit (RELU) layer, a $16N'$-neuron FCL, a RELU layer, and then an FCL with $N'$-neuron, a hyperbolic tangent (tanh) layer and a scaling output layer to map the output space domain into $[-\pi, \pi]^{N'}$. The critic DNNs for each of the two agents of the TD3-DRL network start with $32N'$-neuron input layer concatenating the observation and action inputs, followed by $32N'$-neuron FCL, RELU, $16N'$-neuron FCL, a tanh layer, and finally a 1-neuron fully connected output layer.  To implement the QLL method, an additional fine-tuning layer is appended between the $16N'$-neuron FCL and the output scaling tanh layer. This fine-tuning layer has the same structure as the $16N'$-neuron FCL, but with different learning parameters.
 
 \subsection{Evaluation of Similarity Criterion}
In this part, we examine the effectiveness of the proposed similarity criterion expressed in \eqref{eq:circular_pearson}, which plays a crucial impact on the performance of the adaptive policy transfer technique. 
Fig. \ref{fig:Correlation criteria} illustrates the degree of correlation (ECC) between several typical subarray PDIs. Fig \ref{fig:Correlation criteria}-a shows the overall PDI after training all subarrays of ELPM through the DRL scheme without employing the TL by considering 3-bit phase shifters. The subarray enclosed in a red square is considered a typical student subarray $m^s$ for which we are going to search through all ELPM subarrays to find teachers with high similarity criterion.
Using \eqref{eq:circular_pearson}, the correlation between PDI of the student subarray $m^s$, and those relating to all possible subarrays $m^t\in \M | m^t\neq m^s$ has been calculated for rotation angles ranging from  $0\degree$ to $360\degree$ with the step size of $10\degree$. 
For each rotation angle $\theta$, we are interested in finding teacher candidates having high values of $C(\wbold^{(m^s)},\wbold^{(m^t)}_{\theta})$ in \eqref{eq:ECC}. For example, if we consider a minimum threshold value of 0.5 ($C>0.5$), the set of candidate teachers is found to be non-empty only for $\theta=0\degree$ and $\theta=90\degree$, as depicted in Figs. \ref{fig:Correlation criteria}-b and \ref{fig:Correlation criteria}-c respectively.  
Upon examination of both cases, it is seen that a substantial number of candidates possess a value of C
 that exceeds 0.5. However, when the similarity threshold is escalated from 0.5 to 0.9, there is a marked decrease in the number of candidates for whom C
 surpasses 0.9. This reduction is significantly more pronounced in Fig. \ref{fig:Correlation criteria}-c as compared to Fig. \ref{fig:Correlation criteria}-b. 


 \subsection{Subarray Policy Propagation}
The effectiveness of the proposed policy propagation scheme in increasing the learning speed of the overall SBF algorithm is investigated in this section. 
Figs. \ref{fig:3bit-subarray TRL} and \ref{fig:4bit-subarray TRL} depict the normalized power density relating to a single subarray, per training iteration number, for the case where 3-bit and 4-bit phase shifters are employed respectively. The results have been obtained by averaging the power density of all 100 subarrays of the ELPM at each training iteration. Four scenarios have been considered, including DRL with no TL as in \cite{monemi20246g} denoted by No-TL,  and DRL with policy propagation (PP) according to Algorithm 1 if the ECC is higher/lower than the minimum allowed threshold, denoted respectively by PP-High-ECC and PP-Low-ECC.  The solid blue curve corresponds to the PP-HighECC-QLL where $C^{(m^t)}>C^\mathrm{th_2}$ and QLL is considered, the dashed blue curve corresponds to the PP-HighECC where $C^{(m^t)}>C^\mathrm{th_2}$ with no QLL (conventional adaptive policy reuse),  and the solid red curve corresponds to the PP-LowECC where $C^\mathrm{th_1}< C^{(m^t)}\leq C^\mathrm{th_2}$. Finally, the red dashed line shows the normalized target power density, which is computed by using the exact CSI of all array elements and tuning the beamfocusing matrix coefficients to make the signals of all the elements in phase at DFP.
It is seen in Fig. \ref{fig:3bit-subarray TRL}  that the PP-HighECC-QLL starts with a bias of 34\% of the target and reaches the maximum power of No-TL at only 25k iterations. This corresponds to experiencing a 4-fold increase in convergence speed of PP-HighECC-QLL compared to that of No-TL. In addition to the high increase in the convergence speed, it is also seen that for PP-HighECC-QLL, the asymptotic power density (i.e., after convergence) is 10\% higher than that in the No-TL scenario, which corresponds to a lower BFR at the DFP. As demonstrated in Fig. \ref{fig:4bit-subarray TRL}, it is observed that by increasing the resolution of the phase shifters to 4-bit, the learning bias of PP-HighECC-QLL has risen to 48\%, in addition to a 5-fold increase in the convergence speed compared to that in the No-TL.  Besides, the power has reached 98\% of the target value.
To examine the impact of using the QLL, by comparing the PP-HighECC-QLL and PP-HighECC TL schemes, it is observed that the proposed QLL technique achieves about 20\% enhancement in the training speed relative to the conventional adaptive policy reuse technique as shown in both Figs. \ref{fig:3bit-subarray TRL} and \ref{fig:4bit-subarray TRL}. While the outperformance of employing TL with proper teacher selection schemes is obviously seen in both figures, the PP-LowECC curve reveals that a bad teacher selection results in severe performance degradation, even worse than that of the No-TL scenario.

We have thus far investigated the impact of TL on the performance metrics of an individual subarray. Moving forward, our evaluation will shift to encompass the performance of the entire ELPM SBF scheme, where all subarrays collaborate to generate a focused beam with a spot-like profile. We compare the TL-based SBF scheme with the conventional SBF without TL, in terms of BFR and convergence. Fig. \ref{fig:ELPM policy propagation} shows the result of the implementation of the proposed TL-based scheme in all subarrays of ELPM to form the overall SBF. The ELPM results have also been compared to the ones obtained from a single subarray.
It is observed that after only 30k iterations, the QLL algorithm attains 96\% of the target power density, whereas, in the No-TL scenario, it reaches only 80\% of the power density obtained after 100k iterations. This implies a 300\% enhancement in the learning speed as well as a 20\% improvement in the power density. From the perspective of the size of the focal region, the synergistic combining of all subarrays beamfocusing submatrices to form the overall ELPM beamfocusing matrix leads to a focused beam with the 
BFR reduced from approximately $45$cm to  $6.3$cm. Comparing the results of PP-HighECC-QLL to the No-TL scenario \cite{monemi20246g}, it is also observed that the learning speed is raised about 4 times and the focused power density at the DFP is enhanced by about 20\%.

    


 \begin{figure}
    \centering
        
    \begin{tabular}{c} 
    \includegraphics[width=0.8\linewidth]{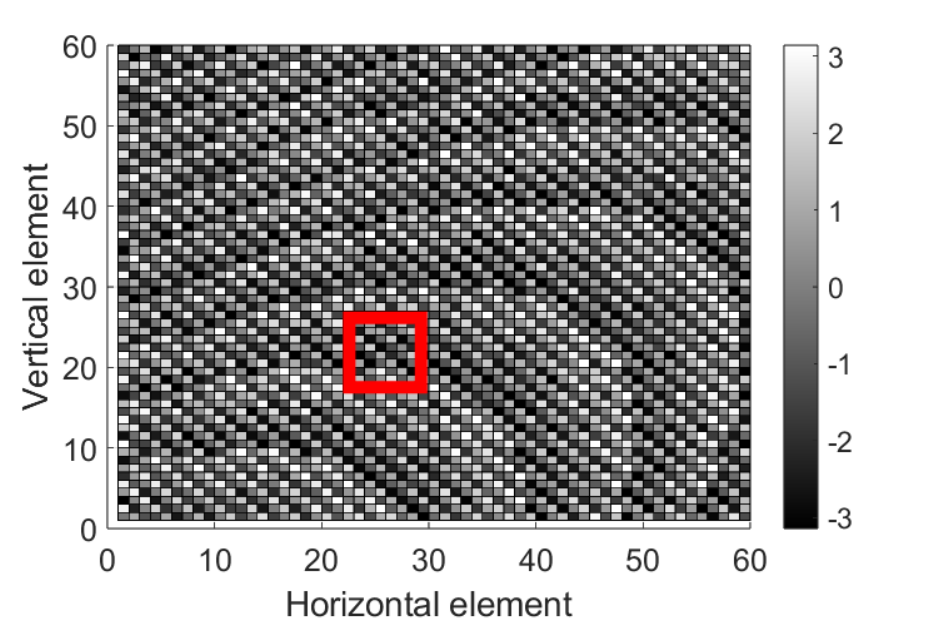}
    \\
    (a)
    \\
    \includegraphics[width=0.8\linewidth]{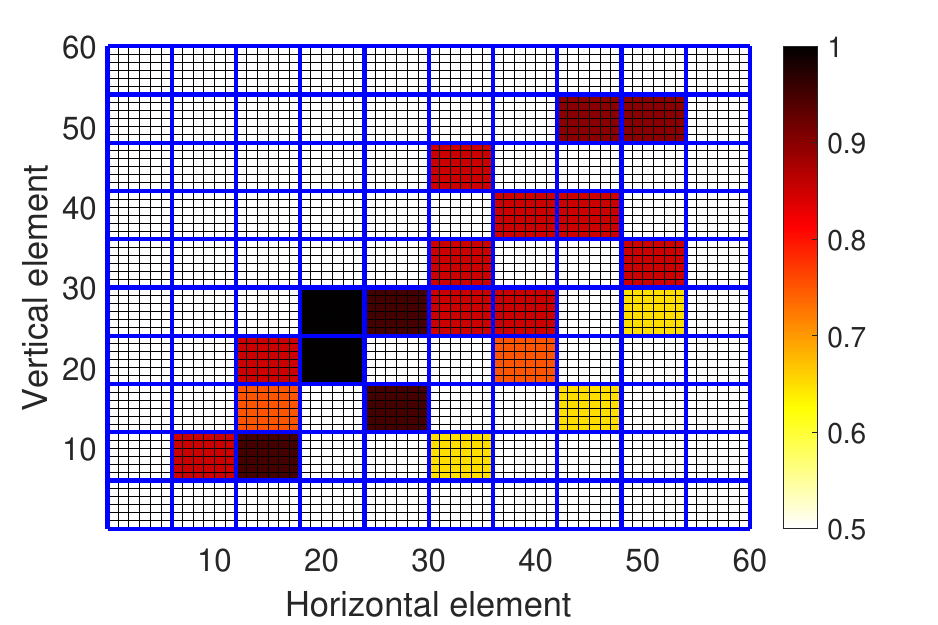}
    \\
    (b)
        \\
    \includegraphics[width=0.8\linewidth]{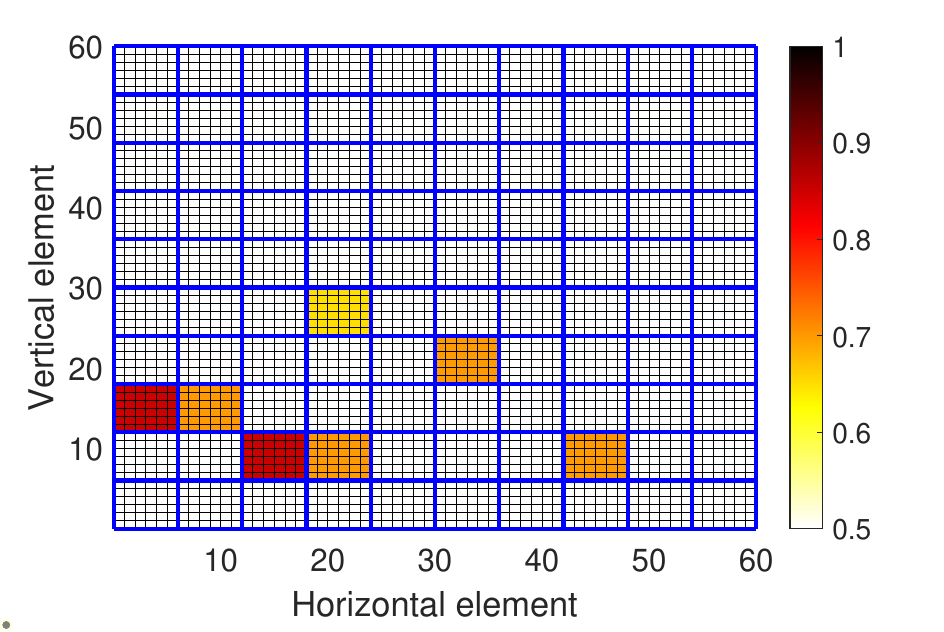}
    \\
     (c)
    \end{tabular}
    \caption{Illustration of the PDI and similarity criterion (a) PDI of ELPM with 3-bit phase shifters and a reference typical subarray (the red box). (b) ECC  between PDIs of  ELPMs subarrays and the reference subarray, for $\theta=0\degree$. (b) ECC  between PDIs of  ELPMs subarrays and the reference subarray, for $\theta=90\degree$.} 

    \label{fig:Correlation criteria}
\end{figure}

 \subsection{DFP Policy Blending}
 In this part, we investigate how selecting proper policy components from the library and applying the blending rule enhances the learning speed for new DFPs.
To evaluate the performance of the proposed policy blending scheme, we set a Monte Carlo simulation scenario as described in the following.
Consider several DFPs each having a random distance from the center of the ELPM with uniform distribution $U(1m,2m)$ for each of which the optimal policy has been obtained through the policy propagation technique and added to the policy library. Once the library is filled with at least $N^\mathrm{lib}$ policies, for a new DFP, the policy blending scheme is run by selecting  $K$-blending component policies whose corresponding DFPs are closest to the new DFP, and then we obtain the optimal policy through Algorithm 2. 
Fig. \ref{fig:policy blending} compares the convergence speed of the policy blending scenarios with $N^\mathrm{lib}\in\{10,50\}$ and $K\in\{1,2,3\}$, as well as the case of No-TL. The blue and orange bars correspond to the results of $N^\mathrm{lib}=10$ and $N^\mathrm{lib}=50$ respectively.
All results have been obtained by averaging from 100 independent Monte Carlo snapshots. It is seen how increasing the number of blending components $K$, as well as the size of the library $N^\mathrm{lib}$ elevates the training speed of the student policy. The first is because a student policy commonly learns better from a higher number of teacher components due to the incorporation of more related knowledge in the training process. The latter is because the growth in the number of DFPs in the library lowers the average distance between the new DFP and the neighboring DFPs $\mathcal{D}^+$. This, potentially raises the ECC of the teacher candidates, resulting in a higher training speed. For example, the best results have been obtained for $K=3$ and $N^\mathrm{lib}=50$, reaching a training speed of about 12k which is 8.3 times faster than the No-TL scenario.

\begin{figure}
		\centering
		\includegraphics [width=254pt]{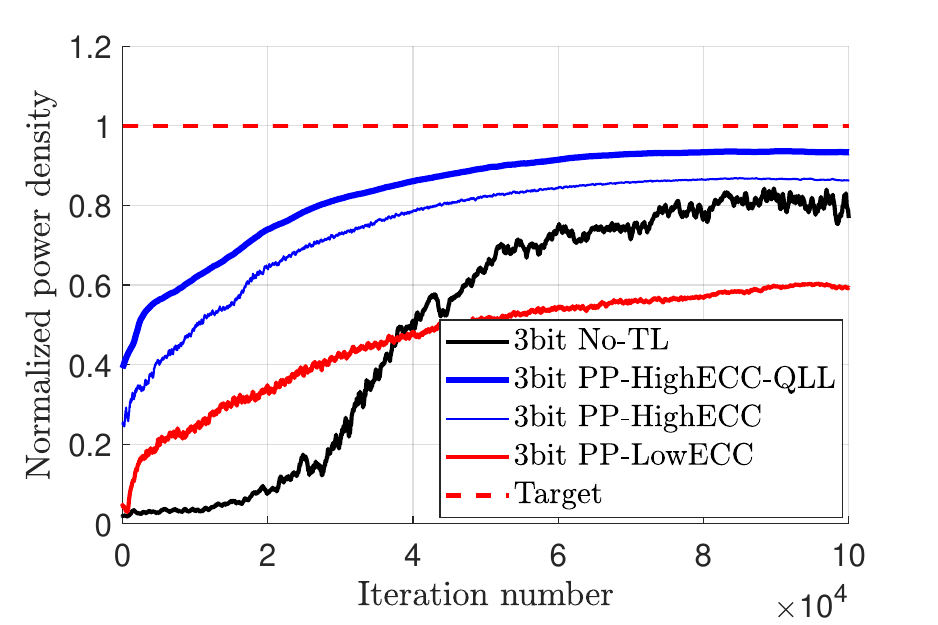} \\
		\caption{The normalized power density per training iteration number for a single ELPM subarray employing a 3-bit phase shifter for four scenarios.
		} 
		\label{fig:3bit-subarray TRL}
\end{figure}

\begin{figure}
		\centering
		\includegraphics [width=254pt]{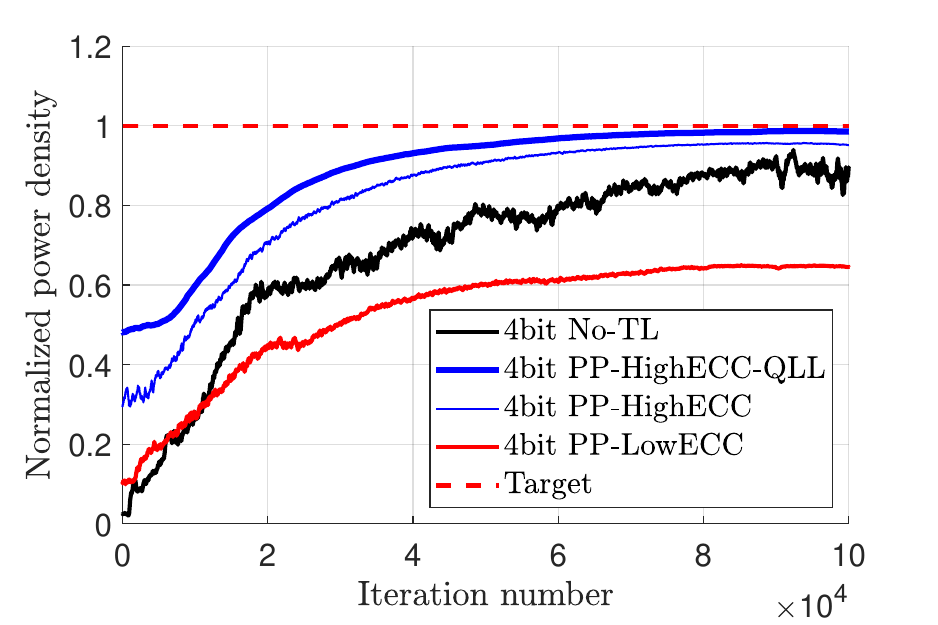} \\
		\caption{The normalized power density per training iteration number for a single ELPM subarray employing a 4-bit phase shifter for four scenarios.
		} 
		\label{fig:4bit-subarray TRL}
\end{figure}

\begin{figure}
		\centering
		\includegraphics [width=277pt]{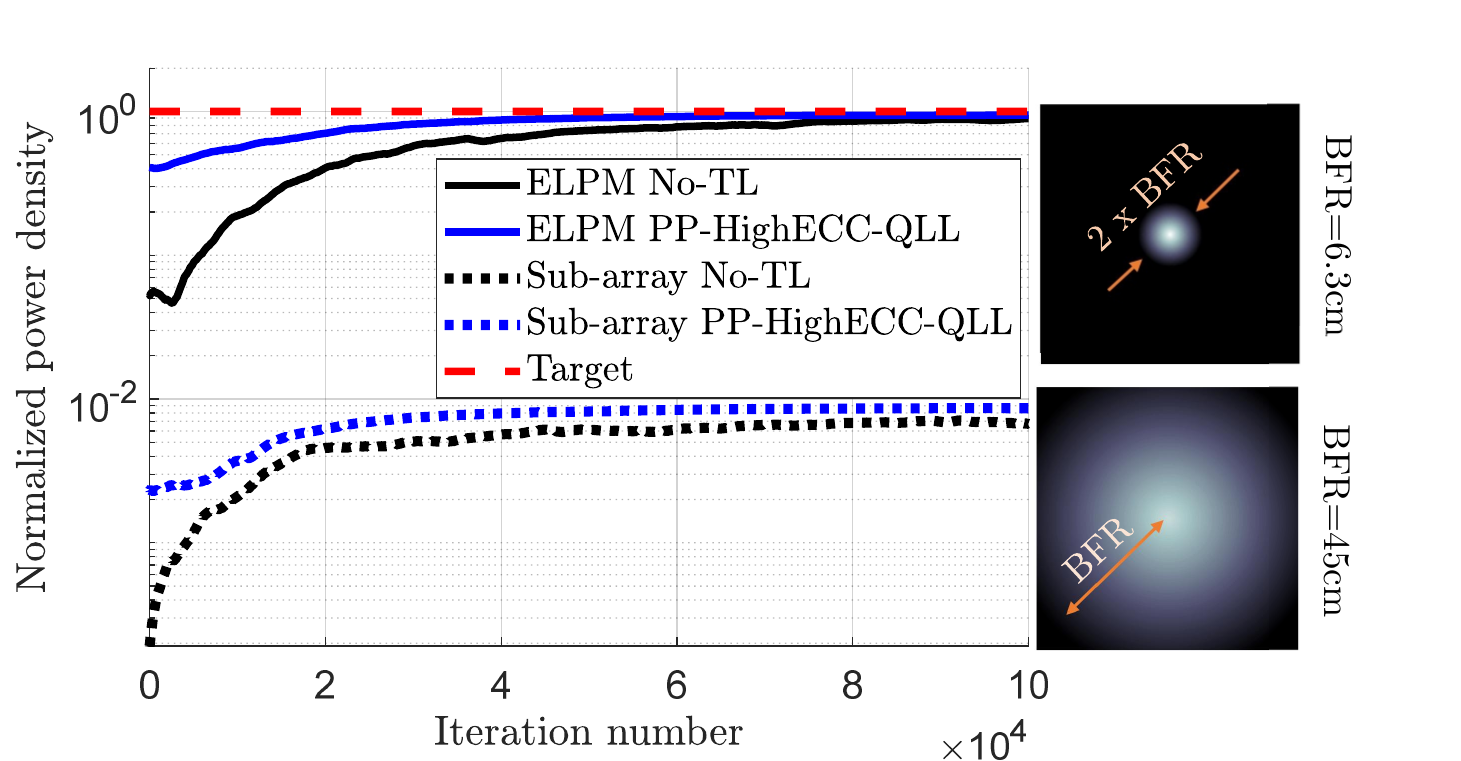} \\
		\caption{Performance comparison of QLL technique applied to all subarrays of the ELPM 
versus that relating to a single subarray in terms of BFR and convergence speed, considering 4-bit phase shifters.  } 
		\label{fig:ELPM policy propagation}
\end{figure}

\begin{figure}
		\centering
		\includegraphics [width=254pt]{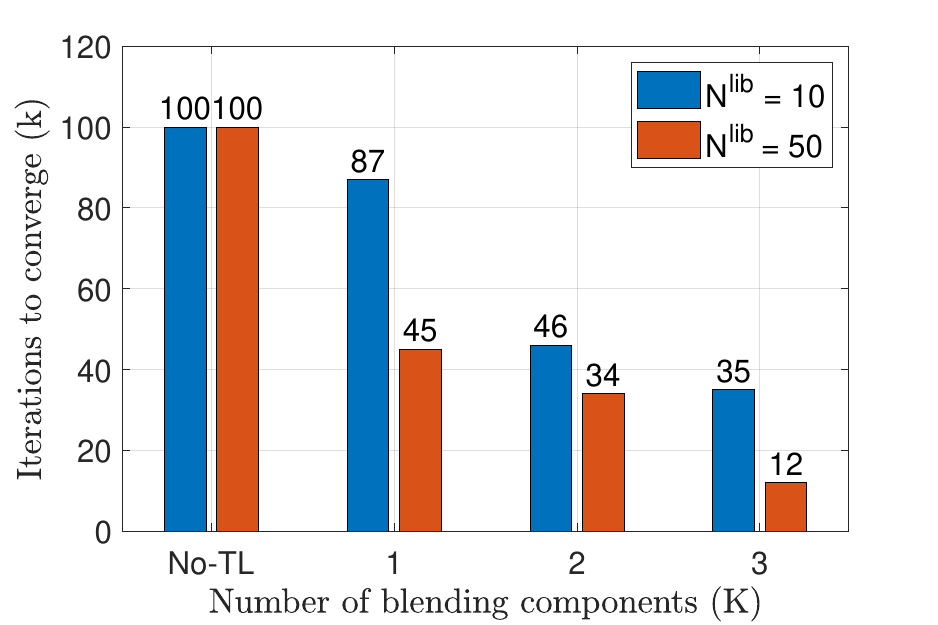} \\
		\caption{Number of iterations required for the convergence of the policy blending algorithm for No-TL as well as the $K$-blending component scenario for $K\in\{1,2,3\}$ and $N^\mathrm{lib}\in\{10,50\}$. 
		} 
		\label{fig:policy blending}
\end{figure}

 \section{Conclusion}

In this paper, we have presented novel TL techniques for near-field SBF using ELPMs, which aim to reduce the training time and enhance the adaptability of smart CSI-independent solutions. We have proposed a subarray policy propagation technique that transfers knowledge between ELPM subarrays' agents by analyzing the phase distribution images. By transferring the knowledge, we managed to accelerate the training process by about 4 times compared to the case where no TL was employed. Furthermore, we introduced the QLL approach as a mechanism to strategically adjust the learning rate of different layers of the DNNs to act as a catalyst and acquire an extra 20\% increase in the training speed. Finally, we dealt with the problem of dynamic DFP management through DFP policy blending, by leveraging the policies relating to previously trained DFPs. 
This technique enables seamless adaptation to new focal points yielding up to an 8-fold reduction in training time. The application of other ML schemes to realize the SBF, as well as employing smart robust SBF mechanisms where the CSI is {\it partially} available remains for future works.


	\bibliographystyle{IEEEtran}
	\bibliography{Mybib}

\begin{biography}[{\includegraphics[width=1in,height
=1.25in,clip,keepaspectratio]{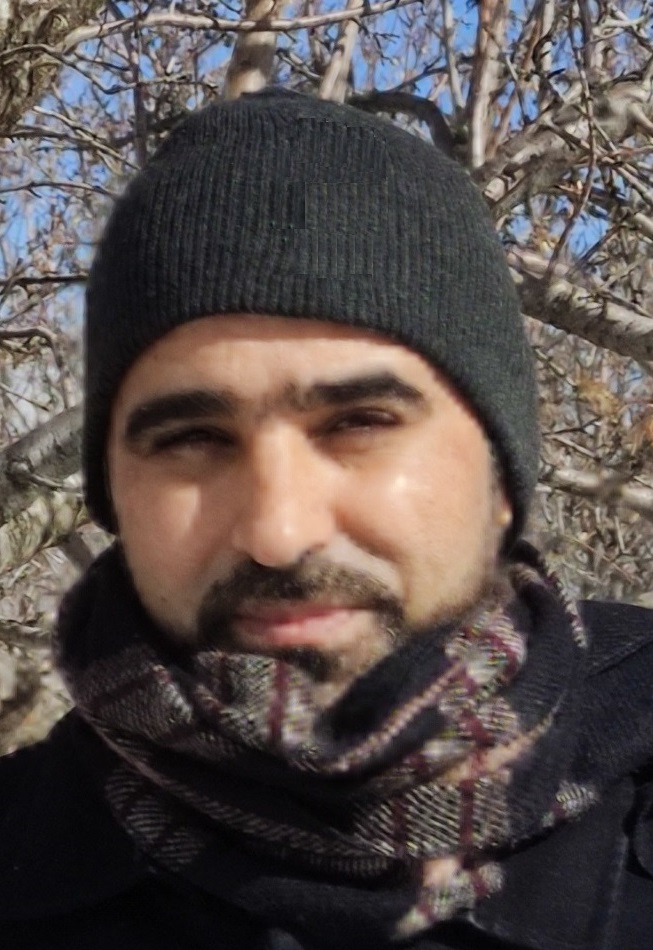}}]{Mohammad Amir Fallah}
		 received the BSc, MSc, and Ph.D. degrees from Shiraz University, Shiraz, Iran, and Tarbiat Modares University, Tehran, Iran, and Shiraz University, Shiraz, Iran, in 2001, 2003 and 2013 respectively, all in electrical and computer engineering. 
    He is an assistant professor with the Department of Engineering, Payame Noor University (PNU), Tehran, Iran, from 2015 till now. His current research interests include antenna and propagation, mobile computing, and the application of machine learning and artificial intelligence in wireless networks.
	\end{biography}

 \vspace{-10pt}
\begin{biography}[{\includegraphics[width=1in,height
=1.25in,clip,keepaspectratio]{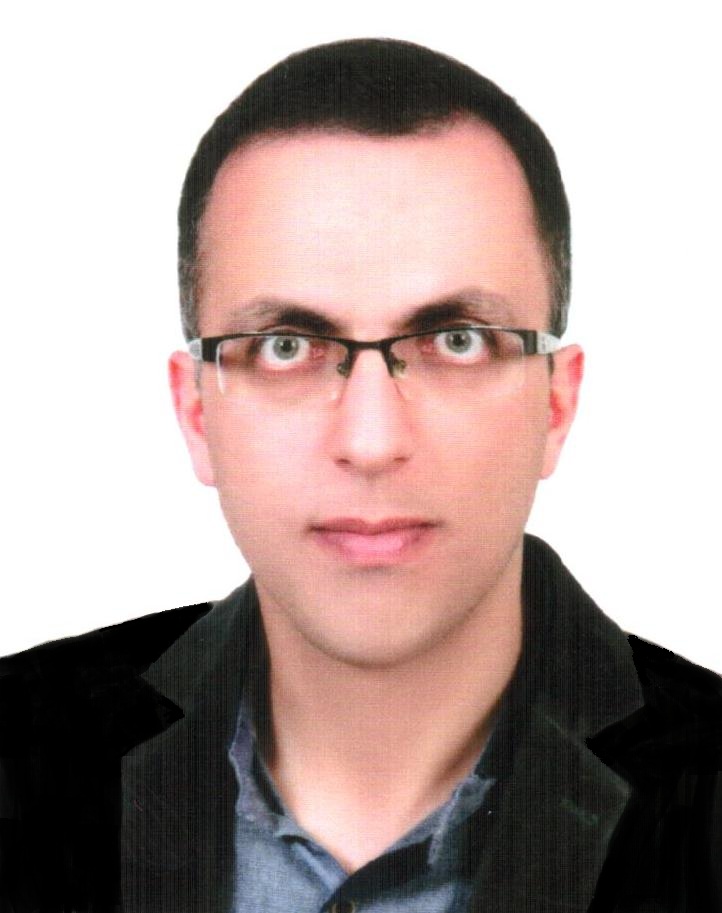}}]{Mehdi Monemi} (Member, IEEE)
		received the B.Sc., M.Sc., and Ph.D. degrees all in electrical and computer engineering from Shiraz University, Shiraz, Iran, and Tarbiat Modares University, Tehran, Iran, and Shiraz University, Shiraz, Iran in 2001, 2003 and 2014 respectively. After receiving his Ph.D., he worked as a project manager in several companies and was an assistant professor in the Department of Electrical Engineering, Salman Farsi University of Kazerun, Kazerun, Iran, from 2017 to May 2023. He was a visiting researcher in the Department of Electrical and Computer Engineering, York University, Toronto, Canada from June 2019 to September 2019. He is currently a Postdoc researcher with the Centre
for Wireless Communications (CWC), University of Oulu, Finland. His current research interests include resource allocation in 5G/6G networks, as well as the employment of machine learning algorithms in wireless networks.
	\end{biography}
 \vspace{-10pt}

\begin{biography}[{\includegraphics[width=1in,height=1.25in,clip,keepaspectratio]{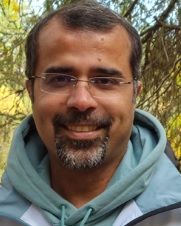}}]{Mehdi Rasti}
		(Senior Member, IEEE) received the
B.Sc. degree in electrical engineering from Shiraz University, Shiraz, Iran, in 2001, and the M.Sc. and Ph.D. degrees from Tarbiat Modares University, Tehran, Iran, in 2003 and 2009, respectively. He is currently an Associate Professor with the Centre
for Wireless Communications, University of Oulu, Finland. From 2012 to 2022, he was with the Department of Computer Engineering, Amirkabir University of Technology, Tehran. From February 2021 to January 2022, he was a Visiting Researcher with the Lappeenranta-Lahti University of Technology, Lappeenranta, Finland.
From November 2007 to November 2008, he was a Visiting Researcher with the Wireless@KTH, Royal Institute of Technology, Stockholm, Sweden. From September 2010 to July 2012, he was with the Shiraz University of
Technology, Shiraz. From June 2013 to August 2013, and from July 2014 to
In August 2014, he was a Visiting Researcher with the Department of Electrical
and Computer Engineering, University of Manitoba, Winnipeg, MB, Canada.
His current research interests include radio resource allocation in IoT, Beyond
5G and 6G wireless networks.
	\end{biography}

 \begin{biography}[{\includegraphics[width=1in,height=1.25in,clip,keepaspectratio]{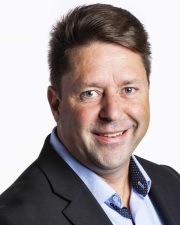}}]{Matti Latva-aho} (Fellow, IEEE) Matti Latva-aho (IEEE Fellow) received his M.Sc., Lic.Tech., and Dr.Tech. (Hons.) degrees in Electrical Engineering from the University of Oulu, Finland, in 1992, 1996, and 1998, respectively. From 1992 to 1993, he was a Research Engineer at Nokia Mobile Phones in Oulu, Finland, after which he joined the Centre for Wireless Communications (CWC) at the University of Oulu. Prof. Latva-aho served as Director of CWC from 1998 to 2006 and was Head of the Department of Communication Engineering until August 2014. He is currently a Professor of Wireless Communications at the University of Oulu and the Director of the National 6G Flagship Programme. He is also a Global Fellow at The University of Tokyo. Prof. Latva-aho has published over 500 conference and journal papers in the field of wireless communications. In 2015, he received the Nokia Foundation Award for his achievements in mobile communications research.
	\end{biography}

\end{document}